\documentclass[a4paper,nosubfigcap,USenglish,thm-restate]{lipics-v2021}

\nolinenumbers
\hideLIPIcs

\usepackage[]{todonotes}
\usepackage{enumerate}
\usepackage{mathtools}
\usepackage{thmtools} 
\usepackage{thm-restate}

\usepackage{cleveref}
\crefformat{footnote}{#2\footnotemark[#1]#3}

\usepackage{tabu}
\usepackage{booktabs}

\usepackage{wrapfig}

\graphicspath{{./figures/}}

\usepackage{xspace}
\newcommand{\dirN}{\textsc{n}\xspace}
\newcommand{\dirE}{\textsc{e}\xspace}
\newcommand{\dirS}{\textsc{s}\xspace}
\newcommand{\dirW}{\textsc{w}\xspace}
\newcommand{\dirNE}{\textsc{ne}\xspace}

\newcommand{\dirNW}{\textsc{nw}\xspace}
\newcommand{\dirWN}{\textsc{wn}\xspace}
\newcommand{\dirSE}{\textsc{se}\xspace}

\newcommand{\dirSW}{\textsc{sw}\xspace}
\newcommand{\dirWS}{\textsc{ws}\xspace}

\newcommand{\diam}{\raisebox{0pt}{{O}}}

\title{Compacting Squares:
  Input-Sensitive In-Place
  Reconfiguration of Sliding Squares}
\author{Hugo A. Akitaya}{University of Massachusetts Lowell, USA}{hugo_akitaya@uml.edu}{https://orcid.org/0000-0002-6827-2200}{}
\author{Erik D.\ Demaine}{Massachusetts Institute of Technology, USA}{edemaine@mit.edu}{https://orcid.org/0000-0003-3803-5703}{}
\author{Matias Korman}{Siemens Electronic Design Automation, USA}{matias_korman@mentor.com}{}{}
\author{Irina Kostitsyna}{TU Eindhoven, The Netherlands}{i.kostitsyna@tue.nl}{https://orcid.org/0000-0003-0544-2257}{}
\author{Irene Parada}{Technical University of Denmark, Denmark}{irmde@dtu.dk}{https://orcid.org/0000-0003-2401-8670}{}
\author{Willem Sonke}{TU Eindhoven, The Netherlands}{w.m.sonke@tue.nl}{https://orcid.org/0000-0001-9553-7385}{}
\author{Bettina Speckmann}{TU Eindhoven, The Netherlands}{b.speckmann@tue.nl}{https://orcid.org/0000-0002-8514-7858}{}
\author{Ryuhei Uehara}{JAIST, Japan}{uehara@jaist.ac.jp}{https://orcid.org/0000-0003-0895-3765}{}
\author{Jules Wulms}{TU Wien, Austria}{jwulms@ac.tuwien.ac.at}{https://orcid.org/0000-0002-9314-8260}{}
\authorrunning{H. Akitaya et al.}

\Copyright{Hugo Akitaya, Erik D. Demaine, Matias Korman, Irina Kostitsyna, Irene Parada, Willem Sonke, Bettina Speckmann, Ryuhei Uehara, and Jules Wulms}

\ccsdesc[100]{Theory of computation~Computational geometry}

\keywords{Sliding cubes, Reconfiguration, Modular robots, NP-hardness}
    
\begin{document}

\maketitle

\begin{abstract}
%   \emph{Sliding squares} is a well-established theoretical model of
%   self-reconfigurable modular robots,
%   analyzed by Dumitrescu and Pach at SoCG 2004.
%   Modules are unit squares connected via shared edges on the square grid.
%   In each move, a module can slide one unit parallel to a shared edge,
%   and then optionally turn a corner by sliding one unit perpendicular to that,
%   provided the remaining modules remain connected via edges.
%   Dumitrescu and Pach proved that this model is universal:
%   any edge-connected configuration of $n$ squares can be reconfigured
%   into any other using $O(n^2)$ sliding moves,
%   and this bound is optimal in the worst case.
%
%A self-reconfigurable modular robot is a highly adaptive robotic system consisting of a number of identical modules arranged on a grid. The modules can move with respect to each other, which enables the robot to change its shape and thus adapt to different environments and purposes.
%
A well-established theoretical model for modular robots in two dimensions are edge-connected configurations of square modules, which can reconfigure through so-called sliding moves.
Dumitrescu and Pach [Graphs and Combinatorics, 2006] proved that it is always possible to reconfigure one edge-connected configuration of~$n$ squares into any other using at most $O(n^2)$ sliding moves, while keeping the configuration connected at all times.
%
%%In this paper we present a simpler and more natural algorithm for self-reconfiguration in the sliding square model, based on the ``compact-and-deploy'' approach.
%% Using the basic principle that well-connected components of modular robots can be transformed efficiently, our algorithm iteratively increases the connectivity within a configuration, to finally arrive at a single solid $xy$-monotone component, before deploying it into the target configuration.
%% %
%% In addition to producing a more natural and coordinated reconfiguration process, our algorithm is input-sensitive and in-place.
%% It requires only $O(\bar{P} n)$ sliding moves to transform one configuration into the other, where $\bar{P}$ is the maximum perimeter of the respective bounding boxes of the two configurations.

  For certain pairs of configurations, reconfiguration may require $\Omega(n^2)$ sliding moves.
  However, significantly fewer moves may be sufficient. We prove that it is NP-hard to minimize the 
  number of sliding moves for a given pair of edge-connected configurations.
  On the positive side we present {Gather\&Compact}, an input-sensitive in-place algorithm 
  that requires only $O(\bar{P} n)$ sliding moves to transform one
  configuration into the other,
  where $\bar{P}$ is the maximum perimeter of the two bounding boxes.
  The squares move within the bounding boxes only, with the exception of at most one square at a time which may move through the positions adjacent to the bounding boxes.
  The $O(\bar{P} n)$ bound never exceeds $O(n^2)$, and is optimal (up to
  constant factors) among all bounds parameterized by just $n$ and~$\bar{P}$.
  Our algorithm is built on the basic principle that
  well-connected components of modular robots can be transformed efficiently.
  Hence we iteratively increase the connectivity within a configuration, to finally arrive at a single solid $xy$-monotone component. 

  We implemented {Gather\&Compact} and compared it experimentally to 
  the in-place modification by Moreno and Sacrist{\'a}n [EuroCG 2020] of the Dumitrescu and Pach algorithm (MSDP).
  Our experiments show that {Gather\&Compact} consistently outperforms MSDP
  by a significant margin, on all types of square configurations.
\end{abstract}

\newpage

\section{Introduction}
\label{sec:introduction}

Self-reconfigurable modular robots~\cite{survey2007} promise adaptive, robust, scalable, and cheap solutions in a wide range of technological areas, from aerospace engineering to medicine. 
Modular robots are envisioned to consist of identical building blocks arranged in a lattice and are intended to be highly versatile, due to their ability to reconfigure into arbitrary forms. 
An actual realization of this vision depends on fast and reliable reconfiguration algorithms, which hence have become an area of growing interest.

%Modular self-reconfigurable robots~\cite{survey2007} consist of many identical modules that can move with respect to each other. In this way, modular robots can change shape and potentially adapt to different environments and purposes. 
One of the best-studied paradigms of modular robots is the \emph{sliding cube model}~\cite{heterogeneous-03}. 
In this model, a robot \emph{configuration} is a face-connected set of cubic modules on the cubic grid. 
The cubes can perform two types of moves, illustrated in two dimensions in Figure~\ref{fig:moves}. 

\begin{figure}[h]
    \centering
    \includegraphics{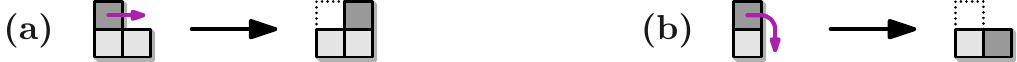}
    \caption{Moves admitted by the sliding cube model: \textbf{(a)} slide, \textbf{(b)} convex transition.}
    \label{fig:moves}
\end{figure}
\noindent
First, a module can \emph{slide} along two face-adjacent cubes to reach a face-adjacent empty grid cell. 
Second, a module $m$ can make a \emph{convex transition} around a module $m'$ to end in a vertex-adjacent empty grid cell. For this second move to be feasible, also the grid cell (not occupied by $m'$) face-adjacent to both the starting and the ending positions must be empty. 
There are several prototypes of modular robots that realize the sliding cube model in 2D~\cite{EMCube,Chiang01,Vertical98}.
Units of multiple other prototypes, 
including expandable and contractible units~\cite{Crystal,Telecube} as well as large classes of modular robots~\cite{metamodule1,metamodule2}, 
can be arranged into cubic \emph{meta-modules} consisting of several units 
such that the meta-module can perform slide and convex transition moves.
Thus, algorithmic solutions in the sliding cube model can be applied to modular robot systems realizing other models.
% with even less free-space requirements than in the sliding cube model. 
% More precisely, for the convex transition it is only required that the ending grid cell is empty. 

%There are several other well-studied models for modular robots, such as the \emph{pivoting cube model}~\cite{heuristics-square,M-blocks}.
Another well-studied model for modular robots is the \emph{pivoting cube model}~\cite{heuristics-square,M-blocks}.
This model strengthens the free-space requirements for each move and has also been realized by some existing prototypes. Akitaya et al.~\cite{pivoting-socg21} showed very recently that in the pivoting cube model in two dimensions it is PSPACE-hard to decide whether it is possible at all to reconfigure one configuration into the other. However, if one allows six auxiliary squares in addition to the input configuration, then there is a worst-case optimal reconfiguration algorithm~\cite{musketeers}. Other models for squares 
relax the face-connectivity condition~\cite{nadia}, 
restrict or enlarge the set of allowed moves~\cite{PSPACE-sliding-corners}, %distributed
or relax the free-space requirements~\cite{squeezing11}. 
% For 3D there is a universal reconfiguration algorithm requiring $O(n^3)$ moves~\cite{pushing-cubes}. 
%our cubes paper \cite{cubes-multimedia}}

In this paper we study the reconfiguration problem for the sliding cube model in two dimensions (the \emph{sliding square model}).
Given two configurations of $n$~unlabeled squares
(each describing the relative positions of squares),
we compute a short sequence of moves that 
transforms one configuration into the other, while preserving edge-connectivity at all times.
Dumitrescu and Pach~\cite{dumitrescu-pushing-squares-2006}
described an algorithm which transforms any two configurations of $n$~squares into each other using $O(n^2)$ moves.
This bound is worst-case optimal: there are pairs of configurations (a horizontal and a vertical line)
which require $\Omega(n^2)$ moves for any transformation.
However, for other pairs of configurations, significantly fewer moves suffice.

We show in Section~\ref{sec:hardness} that it is NP-hard to minimize the number of sliding moves for a given pair of edge-connected configurations. Due the $O(n^2)$ upper bound on the number of moves, the corresponding decision problem is NP-complete.
In Section~\ref{sec:algorithm}, we present a input-sensitive and in-place
algorithm for self-reconfiguration, based on the ``compact-and-deploy'' approach. 
Using the basic principle that well-connected components of modular robots can be transformed efficiently, our algorithm iteratively increases the connectivity within a configuration, to arrive at a single solid $xy$-monotone component, before deploying it into the target configuration. Hence, our algorithm builds the target configuration in such a way, that the lower left corner of the bounding boxes of both configurations are aligned.
Our algorithm is \emph{input-sensitive}:
it requires only $O(\bar{P} n)$ sliding moves to transform one configuration into the other, where $\bar{P}$ is the maximum perimeter of the two bounding boxes.
Our algorithm is also \emph{in-place}: 
only one square at a time is allowed to move through cells edge-adjacent to the respective bounding box.

\begin{figure}[b]
    \centering
    \includegraphics[width=\textwidth]{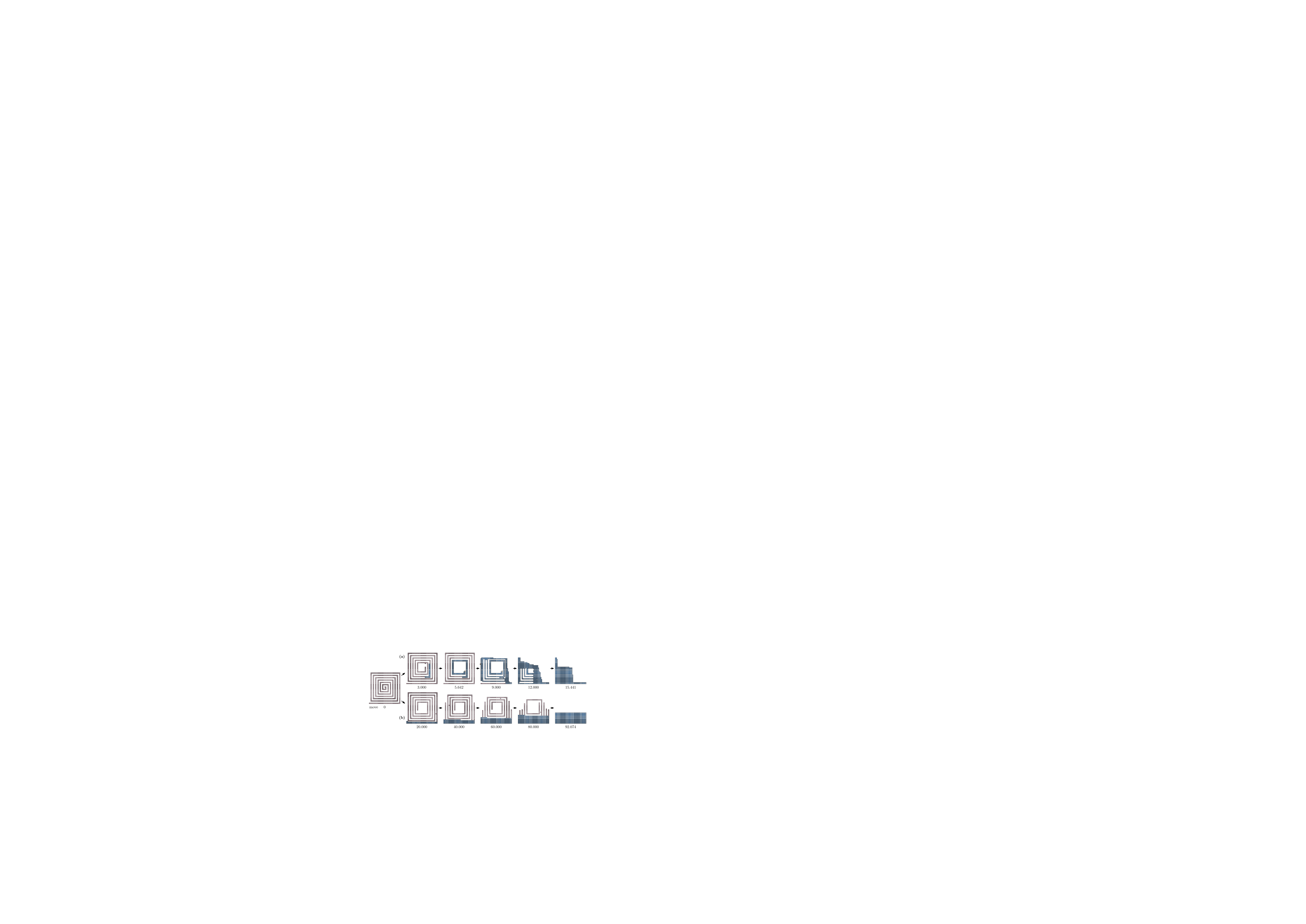}
    \caption{A spiral configuration in a $40 \times 40$ bounding box. \textbf{(a)} {Gather\&Compact}: gathering done after $5.642$ moves; total $15.441$ moves. \textbf{(b)} MSDP~\cite{dumitrescu-pushing-squares-2006,flooding}: total $92.074$ moves.\\ Video: \includegraphics[width=2.5mm]{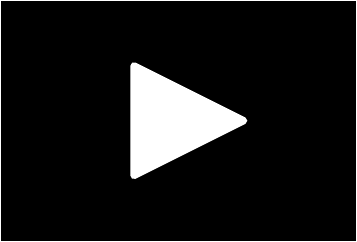} \url{https://tinyurl.com/algaspiral}}
    \label{fig:spiral}
\end{figure}

\subparagraph*{Lower bound.}
Our $O(\bar{P} n)$ bound is optimal (up to constant factors) among all
bounds parameterized by just $n$ and $\bar P$.  Consider any
$x \times y$ rectangle with perimeter $\bar P = 2x+2y$,
where $x \geq y \geq 7$, $x \leq n/2$, and $x \cdot y \geq 6n$.
Place $n/2$ squares at the top of the rectangle (the \emph{lintel}),
which occupy at least one row (because $x \leq n/2$) and at most
$1/6$ of the rows (because $x \cdot y \geq 6n$),
in both the source and target configuration.
Partition the remaining area of this rectangle into thirds
with roughly vertical cuts.
In the source configuration, fill $n/2$ squares in the left third
(column by column, left to right, a \emph{pillar}),
and in the target configuration, fill $n/2$ squares in the right third
(column by column, right to left, a \emph{pillar}).

Our algorithm aligns the lower left corners of the two bounding boxes. In this case it is easy to see that any reconfiguration needs to move $n/2$ squares from the source pillar to the target pillar across the middle third of the configurations,
which has width roughly $x/3$, which is $\Omega(\bar P)$ because $x \geq y$.
It remains to argue that no other alignment of the two bounding boxes reduces the required number of moves asymptotically. Any shift of the bounding boxes by $\Omega(x)$ moving the pillars close to each other, forces $\Omega(n)$ squares in the source and target lintel to move $\Omega(x) = \Omega(\bar P)$. Generally speaking, either the squares in the lintels or the squares in the pillars (or both) have to perform  $\Omega(\bar P n)$ sliding moves in total.

\subparagraph*{Comparison with Dumitrescu and Pach.}
The algorithm by Dumitrescu and Pach~\cite{dumitrescu-pushing-squares-2006} 
constructs a canonical shape from both input configurations. In the original paper this canonical shape is a strip that grows to the right of a right-most square and hence, necessarily, their algorithm always requires $\Theta(n^2)$ moves.
Moreno and Sacrist{\'a}n~\cite{BachelorMoreno,flooding} modify Dumitrescu and Pach to be in-place; their canonical shape is a rectangle within the bounding box of the input. 
For either type of canonical shape the algorithm roughly proceeds as follows. If there is a square which is a leaf in the edge-adjacency graph, then the algorithm attempts to move this square along the boundary towards the canonical configuration. If this leaf square ``gets stuck'' on the way, and hence increases its connectivity, or if there is no leaf in the first place, then the algorithm identifies a 2-connected square on the outside of the configuration which it can move towards the canonical configuration. Hence, if configurations are tree-like (such as the spiral illustrated in Figure~\ref{fig:spiral} left), then each square moves along all remaining squares, for a total of $\Omega(n^2)$ moves (see Figure~\ref{fig:spiral} bottom row). However, the width and the height of this spiral configuration is $O(\sqrt{n})$. Our algorithm gathers $\Theta(\sqrt{n})$ squares from the end of the spiral and then compacts in a total of $O(n \sqrt{n})$ moves. 

The in-place modification by Moreno and Sacrist{\'a}n of Dumitrescu and Pach (henceforth MSDP) has the potential to use fewer than $\Theta(n^2)$ moves in practice. In Section~\ref{sec:experiments} we compare our Gather\&Compact to MSDP experimentally; Gather\&Compact consistently outperforms MSDP by a significant margin, on all types of square configurations.

\section{Hardness of optimal reconfiguration}
\label{sec:hardness}

In this section we sketch the proof of Theorem~\ref{thm:NP-hard-sq}, all details can be found in Appendix~\ref{app:NPhard-sq}.

\begin{restatable}{theorem}{NPhard}
%\begin{theorem} 
\label{thm:NP-hard-sq}
Let $\mathcal{C}$ and $\mathcal{C'}$ be two configurations of $n$ squares each  
and let $k$ be a positive integer. 
It is NP-complete to determine whether
we can transform $\mathcal{C}$ into $\mathcal{C'}$ using at most $k$ sliding moves while maintaining edge-connectivity at all times. 
%\end{theorem}
\end{restatable}

We provide a reduction from \textsc{Planar Monotone 3SAT}. In particular, we start from a rectilinear drawing (see Figure~\ref{fig:monotone3SAT}) of a planar monotone 3SAT instance $\mathcal{C}$ with $n$ variables and $m$ clauses.
We create a problem instance of the reconfiguration problem whose size is polynomial in $n$ and $m$; we show that~$\mathcal{C}$ can be satisfied if and only if the corresponding reconfiguration problem can be solved in at most $66m+24n$ moves. 

\begin{figure}[b]
    \centering
        \includegraphics[page=3,scale=.85]{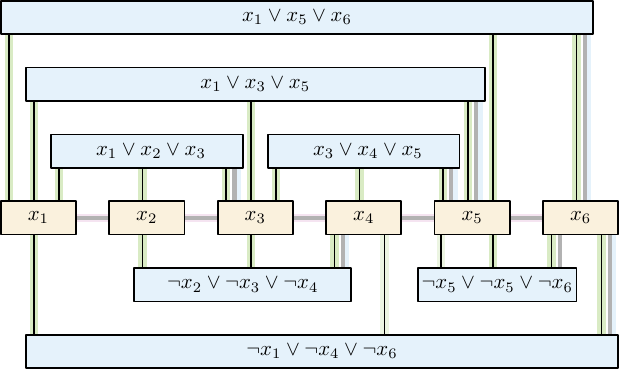}
        \caption{Rectilinear drawing of a \textsc{Planar Monotone 3SAT} instance. Our reduction attaches the variable gadgets horizontally and the clause gadgets %to the variable gadgets 
        next to the rightmost literal in the clause.}
        \label{fig:monotone3SAT}
 \end{figure}
   
 \begin{figure}[tb]
        \centering
        \includegraphics[scale=.55,page=2]{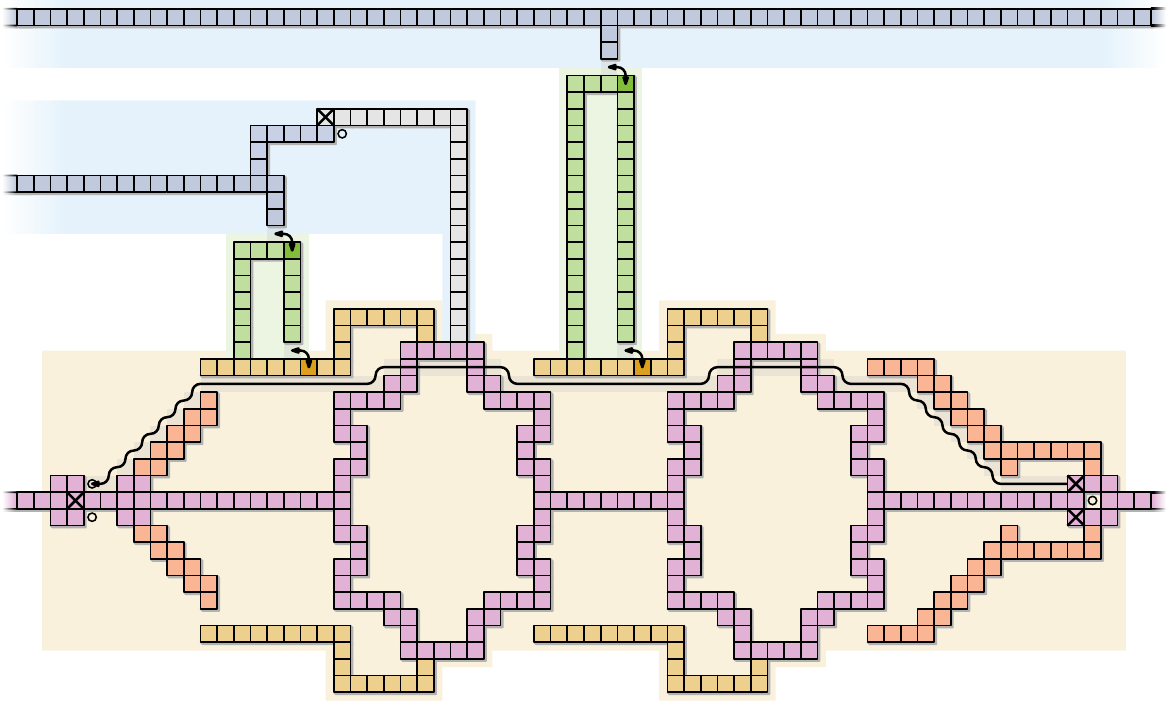} % TODO [ws] fix this figure so it fits the page
        \caption{Overview of the reduction. Background colors match Figure~\ref{fig:monotone3SAT}.}
        \label{fig:reduction}
    %\end{minipage}
\end{figure}

\begin{figure}[tb]
    \centering
    \includegraphics[scale=.55,page=4]{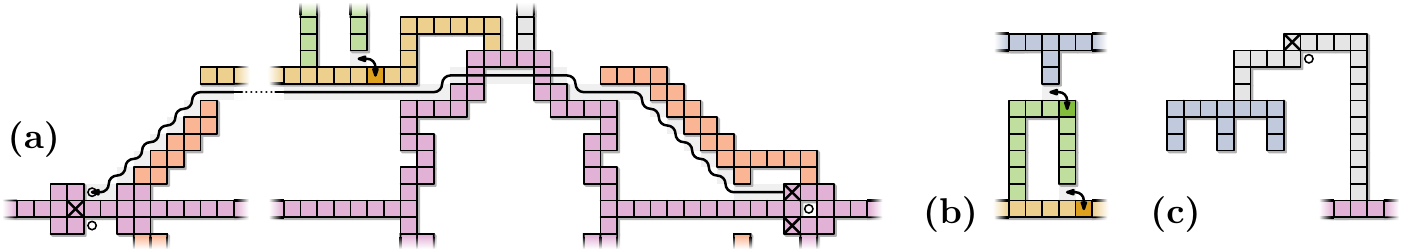}
    \caption{\textbf{(a)} Variable gadget. \textbf{(b)} Wire gadget. \textbf{(c)} Clause gadget.}
    \label{fig:gadgets}
\end{figure}

We replace each variable with a variable gadget, highlighted by an orange-shaded area in Figure~\ref{fig:reduction} and summarized in Figure~\ref{fig:gadgets}a. 
Consecutive variable gadgets are connected by a horizontal line of squares forming a central path of cycles through all variable gadgets (pink squares). 
%With these connections our reduction has a central path of cycles that goes through all variable gadgets (pink squares). 
More precisely, the gadget associated with variable $x_i$ has $k_i$ \diam-shaped cycles in the path of cycles, where~$k_i$ is the number of times that the variable $x_i$ appears in $\mathcal{C}$. 
Each \diam-shaped cycle has two {\em prongs} (yellow squares). 
The spacing between gadgets is large enough for different variable gadgets not to interact in an optimal reconfiguration process. 

The source and target configurations in the variable gadgets are very similar. 
The only difference is that the positions marked with $\times$ in Figure~\ref{fig:gadgets}a must be emptied and the positions marked with~$\circ$ must be occupied. 
Using an earth movers argument one can argue that a square must be transferred from the right of the gadget to the left, and the minimum number of moves required to do so is the horizontal distance between the positions, in this case $20k_i + 20$. 
Following the path in black in  Figure~\ref{fig:gadgets}a (or the symmetric one) achieves that bound. 
The remaining required changes need four additional moves. 
In Appendix~\ref{app:NPhard-sq} we argue that essentially there are only two ways to achieve the minimum number of moves: 

\begin{restatable}{lemma}{varlem}
%begin{lemma}
\label{lem:variable}
The reconfiguration of variable gadget $x_i$ 
needs at least $20k_i+24$ many moves. 
Moreover, the only way to achieve 
that number moves implies transferring one of the right $\times$ squares to the left $\circ$ position at the same height along the path shown in Figure~\ref{fig:reduction} or along the horizontally symmetric one.
%\end{lemma}
\end{restatable}

Lemma~\ref{lem:variable} shows that to reconfigure using as few moves as possible we must transfer one of the two $\times$ squares on the right to a $\circ$ position on the left. During the process, the $\times$ square creates cycles involving the central path of cycles and either all upper or lower prongs. 

The clause gadget mainly consists of a set of squares forming a \emph{pitchfork} ($\pitchfork$) shape (blue squares in Figure~\ref{fig:gadgets}c). 
The pitchfork has three \emph{tines} consisting of two squares each. 
Each tine corresponds to a literal in the clause. 
We add a path of squares connecting the $\pitchfork$ shape to the central path of cycles so that the source configuration is connected (gray squares). 
%As in the variable gadget, 
Most squares are in both the source and the target configurations. The only exception is one square (marked with $\times$) that wants to be transferred to a nearby position (marked with $\circ)$. 
However, the move is initially not possible as it would disconnect the $\pitchfork$ part.

The wire gadgets are connected to the variable gadgets and part of them is placed very close to each of the tines of a pitchfork. 
A wire gadget is a path of squares that form a $\sqcap$ shape for positive literals and a $\sqcup$ shape for negative ones (see Figure~\ref{fig:gadgets}b, green squares). 
Each wire gadget is attached to a different prong of the corresponding variable gadget. 
This associates each literal in a clause to a wire gadget and a prong (note that there can be spare prongs). 
The goal is to allow creating a different connection between the $\pitchfork$ of a clause gadget and the central path of cycles in two moves as long as the prong is in a cycle.

%The clause and wire gadgets are also placed according to the rectilinear drawing of the planar monotone 3SAT instance. 
To avoid interference between different gadgets we place the clause gadgets at different heights and make the vertical separations between gadgets large enough. 
With that we prove (in Appendix~\ref{app:NPhard-sq}) the following lemma:

\begin{restatable}{lemma}{clauselem}
%\begin{lemma}
\label{lem:clause}
A clause gadget needs at least six moves to be reconfigured and it can be reconfigured with six moves if and only if 
a prong associated to a literal in the clause is part of a cycle. 
%end{lemma}
\end{restatable}

The six moves required by a clause gadget are in fact additional to the $20k_i+24$ moves required to reconfigure the gadget for variable $x_i$ and to the six moves required by any other clause gadget. 
If we  allow  only the minimum number of moves per gadget, 
Lemma~\ref{lem:variable} forces that in each variable gadget either the upper or the lower prongs become part of cycles. 
Moreover, Lemma~\ref{lem:clause} requires that for each clause there is a prong associated to a literal in the clause that becomes part of a cycle as part of the reconfiguration of the variable gadgets. 
This implies that if a reconfiguration sequence exists, the 3SAT instance must be satisfiable. 
In the other direction, if the 3SAT instance is satisfiable, then we show how to order the moves carefully to reconfigure with the minimum number of moves required.
% In the other direction, if the 3SAT instance is satisfiable, a careful order in which the moves are done allows to reconfigure the construction in the minimum number of moves required.
% This is the essence of the proof of the following lemma, which finishes the proof of Theorem~\ref{thm:NP-hard-sq}.

\begin{restatable}{lemma}{finallem}
%\begin{lemma} 
A \textsc{Planar Monotone 3SAT} instance can be solved if and only if the corresponding  reconfiguration problem instance can be solved using $66m+24n$ sliding moves.
%\end{lemma}
\end{restatable}

\section{Input-sensitive in-place algorithm}
\label{sec:algorithm}

To describe our input-sensitive reconfiguration algorithm, we first need to introduce the following definitions and notations.
Let $\mathcal{C}$ be an edge-adjacent configuration of squares on the square grid and let $G$ be the \emph{edge-adjacency graph} of $\mathcal{C}$. In~$G$ each node represents a square and two nodes are connected by an edge, if the corresponding squares are edge-adjacent.
With slight abuse of notation we identify the squares and the nodes in the graph. A square $s \in \mathcal{C}$ is called a \emph{cut square} if $\mathcal{C} \setminus \{ s \}$ is disconnected. Otherwise, $s$ is called a \emph{stable square}.

A \emph{chunk} is any inclusion-maximal set of squares in~$\mathcal{C}$ bounded by (and including) a simple cycle in~$G$ of length at least $4$ (its \emph{boundary cycle}), including any edge-adjacent squares with degree-1 in~$G$ (its \emph{loose squares}).
A chunk constitutes a well-connected component that can be efficiently transformed towards an $xy$-monotone configuration.

A \emph{link} is a connected component of squares which are not in any chunk. A \emph{connector} is a chunk square edge-adjacent to a square in a link or in another chunk. By definition a connector is always a cut square. The \emph{size} of a chunk~$C$ is the number of squares contained in~$C$ (including its boundary cycle and any loose squares).

Figure~\ref{fig:component-tree}a shows an example configuration with its chunks, links, connectors, and cut\,/\,stable squares marked. Note that a square can be part of two chunks simultaneously, in which case it must be a connector (for example, see the leftmost connector in Figure~\ref{fig:component-tree}a). A chunk can contain both cut and stable stables.

\begin{figure}[t]
    \centering
    \includegraphics{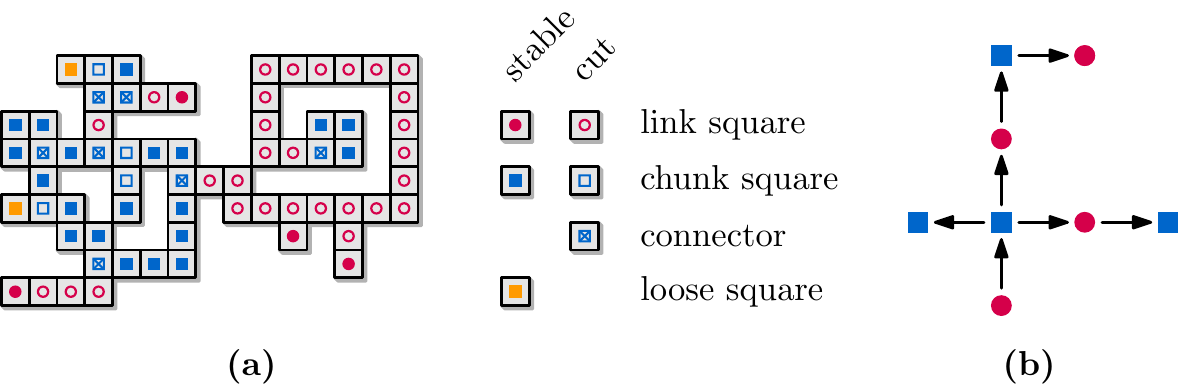}
    \caption{\textbf{(a)} A configuration~$\mathcal{C}$. \textbf{(b)} The component tree~$T$.}
    \label{fig:component-tree}
\end{figure}

The \emph{component tree} $T$ of $\mathcal{C}$ has a vertex for each chunk or link and an edge $(u, v)$ iff the chunks\,/\,links represented by $u$ and~$v$ 
have edge-adjacent squares or share a square (in the case of two adjacent chunks). The component tree is rooted at the component that contains the leftmost square in the bottom row, the \emph{root square}, of~$\mathcal{C}$ (see Figure~\ref{fig:component-tree}b). 
% To see that the component tree~$T$ is indeed always a tree, observe that if to the contrary $T$ contained a cycle, that would correspond to a cycle in $G$ as well. This cycle would form (part of) a single chunk represented as one vertex in~$T$.
If a chunk is a leaf of $T$ then we call it a \emph{leaf chunk}.

A \emph{hole} in~$\mathcal{C}$ is a finite maximal vertex-connected set of empty grid cells. The infinite vertex-connected set of empty grid cells is the \emph{outside}. 
If a chunk~$C$ encloses a hole in~$\mathcal{C}$, we say that $C$ is \emph{fragile}. 
Otherwise, we say that $C$ is \emph{solid}.
The \emph{boundary} of the configuration~$\mathcal{C}$ is the set of squares which are vertex-connected to any grid cell on the outside.
The \emph{boundary} of a hole~$H$ is the set of squares vertex-connected to any grid cell in~$H$. 
Note that the boundary of a hole is edge-connected.
We can construct the component tree in $O(n)$ time by walking along the boundary of~$\mathcal{C}$.

\begin{wrapfigure}[7]{r}{0.2\textwidth}
    \raggedleft
    \vspace{-\intextsep}
    \includegraphics{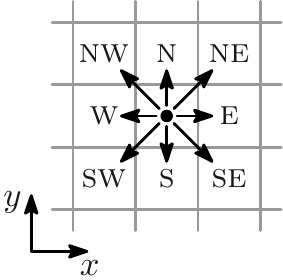}
\end{wrapfigure}
Consider now the bounding box $B$ of $\mathcal{C}$ on the square grid. We refer to the bottom-most left-most grid cell inside $B$ as the \emph{origin}. We say that $P$ is the perimeter of $B$, and hence any square in $\mathcal{C}$ can be connected to the origin by an $xy$-monotone path of at most $P/2$ squares.

Let~$c = (x, y)$ be a grid cell. We use compass directions (\dirN, \dirNE, \dirE, etc.) to indicate neighbors of~$c$. When we use grid coordinates, we assume the usual directions (the $x$-axis
increases towards \dirE and the $y$-axis increases towards \dirN, so the \dirN-neighbor of~$c$ is $(x, y + 1)$). Similarly, we indicate slide moves using compass directions (`a \dirW-move') and convex transitions using a sequence of two compass directions (`a \dirWS-move': a movement toward \dirW followed by a movement towards \dirS).

\subparagraph*{Algorithmic outline.}
In the first phase of our algorithm we ensure that the leaves of the component tree~$T$ are sufficiently large and well-connected.
Specifically, we gather squares from the leaves of~$T$ until each leaf is a chunk of size at least $P$. In Section~\ref{sec:gathering} we explain how to grow chunks using at most $O(P)$ moves per each square that was moved. % which is added to a chunk. 
During this process, the final position of each square is chosen inside bounding box $B$, but squares can move through the layer of grid cells adjacent to $B$.

After gathering, all leaves are \emph{heavy chunks} of size at least $P$. Our goal is now to make each leaf chunk contain the origin, while ensuring that all squares remain part of their chunk.
A heavy chunk $C$ contains sufficient number of squares to be transformed into a chunk containing both the connector of $C$ and the origin: we can connect the connector with the origin by an $xy$-monotone path of at most $P/2$ squares; two such paths, which are disjoint, form a new boundary cycle for $C$.
We do not explicitly construct these two paths, but instead we \emph{compact} the configuration by filling holes and using lexicographically monotone movement towards the origin for squares in heavy leaf chunks. In Section~\ref{sec:compaction} we explain the details of the compaction algorithm and prove that it leaves us with a solid $xy$-monotone component. 

During compaction each square in a leaf chunk makes only lexicographic monotone moves towards the origin while staying inside~$B$: \dirS- and \dirW-moves (slides), as well as \dirSW-, \dirWS-, \dirNW-, and \dirWN-moves (convex transitions). In some cases, a square in the leftmost column or bottom row can exit $B$, and move along the bounding box to enter the same column/row again closer to the origin. This is the only time when a non-lexicographic monotone move is used, and every square can perform it at most $O(P)$ times. Hence the total number of moves during compacting is $O(Pn)$.

When compacting, every (heavy) leaf chunk will eventually contain a square at the origin. This means that the whole configuration becomes a single chunk, as all leaves of the component tree have merged into a single component. Therefore, once compacting has finished we arrive at an $xy$-monotone configuration that fits inside $B$. If at any point during this process the configuration becomes $xy$-monotone, then we simply stop. In particular, if the configuration is $xy$-monotone at the start, for example squares in only a single row or column, then we do not have to gather or compact, even though there are no heavy chunks. See the top row of Figure~\ref{fig:spiral} for a visual impression of our algorithm.

In the special case that the input configuration $\mathcal{C}$ contains less than $P$ squares, we first ensure that $\mathcal{C}$ contains the origin and then execute the gathering and compaction steps as before. The number of moves is trivially bounded by $O(Pn)=O(P^2)$, see Appendix~\ref{sec:light}.

Finally, in Section~\ref{sec:canonical} we show how to convert any $xy$-monotone configuration into a different $xy$-monotone configuration with at most $O(\bar{P}n)$ moves, where $\bar{P}$ is the maximum perimeter of the bounding boxes of source and target configurations. Thus, since all moves are reversible, we can transform the source into the target configuration via this transformation.
All omitted proofs can be found in Appendix~\ref{sec:omitted}.
%Furthermore, in Appendix~\ref{sec:canonical} we explain how to convert any configuration into a different $xy$-monotone configuration with at most $O(\bar{P}n)$ moves, where $\bar{P}$ is the maximum perimeter of the bounding boxes of start and end configuration. Since all moves are reversible we can transform the start into the end configuration via this transformation. 

\begin{theorem}
Let $\mathcal{C}$ and $\mathcal{C'}$ be two configurations of $n$ squares each, let $P$ and $P'$ denote the perimeters of their respective bounding boxes, and let $\bar{P} = \max \{ P, P' \}$. We can transform $\mathcal{C}$ into $\mathcal{C'}$ using at most $O(\bar{P}n)$ sliding moves while maintaining edge-connectivity at all times. 
\end{theorem}
\begin{proof}
For any two configurations $\mathcal{C}$ and $\mathcal{C'}$ of $n$ squares each, we can apply gathering and compacting to find $xy$-monotone configurations $M$ in $O(Pn)$ and $M'$ in $O(P'n)$ moves, for~$\mathcal{C}$ and $\mathcal{C'}$ respectively. If we want to transform $\mathcal{C}$ into $\mathcal{C'}$, we first gather and compact $\mathcal{C}$ into~$M$, transform $M$ into $M'$ in $O(\bar{P}n)$ moves, and proceed by reversing the sequence of steps for $\mathcal{C'}$ to get configuration~$\mathcal{C'}$. In Sections~\ref{sec:gathering},~\ref{sec:compaction} and~\ref{sec:canonical} we show that gathering, compacting and transforming $xy$-monotone configurations require the appropriate number of moves, such that the total number of moves is $O(Pn +\bar{P}n + P'n) = O(\bar{P}n)$.
\end{proof}

\subsection{Gathering}
\label{sec:gathering}

In this section we show how to gather squares from the leaves of the component tree $T$ until we create a chunk of size at least $P$ that is a leaf of $T$. In the following, let $s$ be a connector or a cut square in a link. By definition, $s$ lies on the boundary of $\mathcal{C}$.
Since $s$ is a cut square, removing $s$ from $\mathcal{C}$ results in at least two connected components.
One of these components contains the root of $T$.
We say that the other (up to three) components are \emph{descendants} of $s$.
Let~$D$ be the set of squares in the descendant components of~$s$. We say that the \emph{capacity} of~$s$ is $|D|$, and that $s$ is \emph{light} if its capacity is less than~$P$ and \emph{heavy} otherwise.

\begin{lemma}
    \label{lem:light-square}
    Let~$s$ be a light square with descendant squares~$D$. Then $s$ can be made part of a chunk with a sequence of $O(P)$ moves by squares in~$D$. The final position of~$s$ is within bounding box~$B$.
\end{lemma}
\begin{proof}
    Observe that there are at most two empty cells $e_1$ (and $e_2$) neighboring $s$, such that moving squares there results in a chunk component containing $s$. Cell $e_1$ (and $e_2$) can be chosen such that they lie inside bounding box $B$: these cells must exist since always at least three neighboring cells are inside $B$, unless $B$ is a single row/column and the configuration was already $xy$-monotone. If such cells are already occupied by squares then $s$ is already in a chunk.
    We argue that we can move squares from the descendant components of $s$ into these empty cells inside the bounding box with at most $O(P)$ moves.
    Once this is accomplished, we repeat the process in the descendant components, for the next light square of maximal capacity, until no light squares remain in the component tree below the chunk containing $s$.
    
    Let $D' \subseteq D$ be a subset of boundary squares in the descendants of $s$ of the subconfiguration $D \cup \{s\}$.
    Select an arbitrary stable square $a \in D'$.
    Such a square exists because of the following: if there is a link component in $D$ that is a leaf in the component tree, then its degree-1 node is stable; and if there is a chunk component in $D$ that is a leaf in the component tree, then an extremal square of the chunk in one of the \dirNE, \dirNW, \dirSE, or \dirSW directions is stable (only one of them can be a connector square).
    
    Consider moving $a$ along the boundary of $D$ towards $e_1$.
    Let $E_a$ be the set of cells that $a$ needs to cross to reach $e_1$.
    If $E_a$ is empty, then we simply move $a$ to $e_1$ (see Figure~\ref{fig:gathering}a), and repeat the procedure for $e_2$ (if it exists).
    In this case, $a$ takes $O(P)$ moves to get to $e_1$.
    
    Now consider the case where $E_a$ is not empty.
    Let $b$ be the first square in $E_a$ on the way from $a$ to $e_1$; let $c$ be the square in $E_a$ that is just before and edge-adjacent to $b$.
    As $b$ is not part of the boundary along which the path from $a$ to $e_1$ is considered, it must be vertex-adjacent to a square that is on that part of the boundary.
    
\begin{figure}[b]
    \centering
    % \includegraphics[page=3]{figures/gathering}
    % \hfill
    % \includegraphics[page=2]{figures/gathering}
    % \hfill
    \includegraphics[page=1]{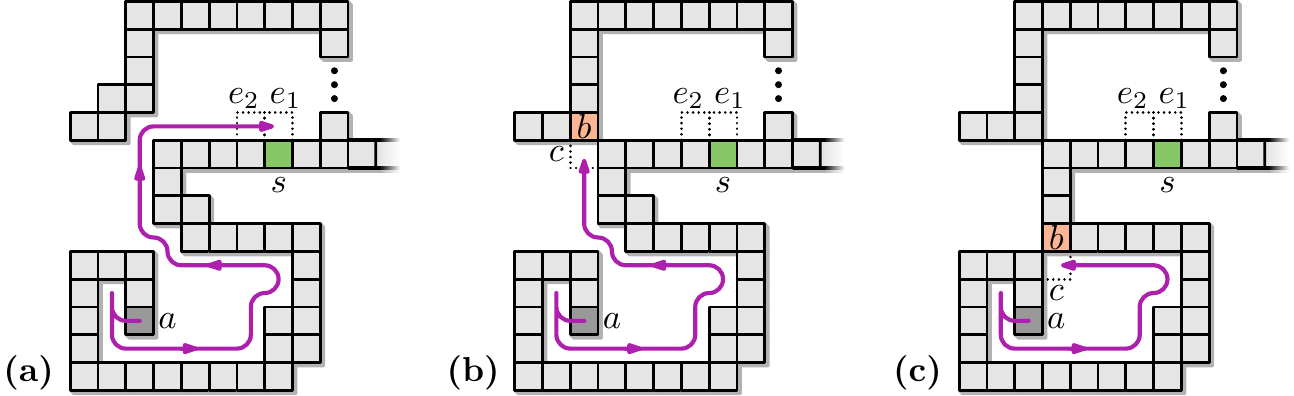}
    \caption{Light square $s$ (green); filling cells $e_1$ and $e_2$ makes $s$ part of a chunk; a stable square $a$ (dark grey) moves towards $e_1$ along the boundary. \textbf{(a)} $a$ reaches $e_1$. \textbf{(b)} square $b$ (brown) part of a component outside of $D$, moving $a$ to $c$ creates a chunk containing $s$. \textbf{(c)} square~$b$ part of a component in $D$, moving $a$ to $c$ creates a hole; its inner boundary will not be traversed again.}
    \label{fig:gathering}   
\end{figure}
    
    There can be two cases: either $b\not\in D$ or $b\in D$.
    In the first case, moving $a$ to $c$ merges a component in $D$ with some component outside of $D$ (see Figure~\ref{fig:gathering}b).
    Thus, a chunk is created that contains $s$, resulting in $s$ no longer being a light square.
    
    In the second case, when $b\in D$, moving $a$ to $c$ creates a chunk within $D$ (see Figure~\ref{fig:gathering}c).
    In this case we select another arbitrary stable square $a'$ in the new subconfiguration $D\cup \{s,c\} \setminus a$, and repeat the procedure.
    Observe that the empty squares traversed by $a$ are now part of a hole in the new chunk.
    Thus the path from $a'$ to $e_1$ does not overlap with the path taken by $a$.
    Let $\{a_0,a_1,a_2,\dots\}$ be the sequence of such stable squares chosen by our algorithm as candidates to be moved to $e_1$.
    For any square $a_i$, its path along the boundary to $e_1$ does not overlap with any of the cells traversed by all $a_j$ with $j<i$.
    Thus, there is some $k$ such that $a_k$ either reaches $e_1$, or it merges the components within and outside of $D$ into a chunk containing $s$.
    The total number of empty cells traversed by all squares $a_i$ ($0\le i\le k$) is~$O(P)$.
    
    We repeat the above procedure to fill~$e_2$. Note that the path taken by~$a$ (and~$a'$) may not always be inside bounding box~$B$. When this path exits~$B$, it will always stay adjacent to cells in $\mathcal{C}$, and hence it will use only the single row/column of cells adjacent to~$B$.
\end{proof}
Using the procedure described in the lemma above, we can iteratively reduce the number of light squares and obtain a component tree where all leaves are chunks of size at least $P$.

\begin{restatable}{lemma}{gatheringLemma}
    \label{lem:gathering}
    There is a sequence of $O(Pn)$ moves which ensures that all leaves in the component tree are chunks of size at least $P$.
\end{restatable}

\subsection{Compacting}
\label{sec:compaction}

After the gathering phase, each leaf of the component tree~$T$ is a heavy chunk, that is, each leaf is a chunk of size at least~$P$. In this section, we describe how we compact the configuration in order to turn it into a \emph{left-aligned histogram} containing the origin. A left aligned histogram has a vertical base, and extends only rightward. Our procedure uses three types of moves: LM-moves, corner moves, and chain moves, which we discuss below. We iteratively apply any of these moves. The correctness of the algorithm does not depend on the order in which moves are executed, as long as the moves are \emph{valid} (defined below). Our implementation assigns priorities to the various types of moves and chooses according to these priorities whenever multiple moves are possible (see Section~\ref{sec:experiments}).

\subparagraph*{LM-moves.}
We say that a move is a \emph{lexicographically monotone move} (\emph{LM-move} for short) if it is either an \dirS- or \dirW-move (slides), or an \dirSW-, \dirWS-, \dirNW-, or \dirWN-move (convex transitions). Note that an LM-move will never move a square to the east. Squares can move to the north, but only when they also move to the west. Hence, if a square starts at coordinate $(x, y)$, and we perform a series of LM-moves, it stays in the region $\{ (x', y') \mid x' \leq x \wedge y' \leq x - x' + y \}$.

% [ws] we're defining boundary cycle in the introduction now
%\item[chunk boundary cycle] maximal boundary cycle of the chunk; the property that we want to maintain is that the chunk consists of the maximal cycle, a few loose squares outside, and whatever structure is inside the cycle; when we talk about moves not breaking two-connectivity of the chunk, we actually mean preserving the above structure of the chunk

Let~$s$ be a square in a heavy leaf chunk~$C$ of~$\mathcal{C}$, and consider an LM-move made by~$s$. We say that this move is \emph{valid} if $s$ stays inside bounding box~$B$, and all squares $s' \in C$ are still in a single chunk after the move. While compacting, we allow only valid LM-moves, that is, we allow each chunk to grow, but a chunk can never lose any squares.

% \begin{figure}[ht]
%     \centering
%     \includegraphics{figures/loose-squares}
%     \caption{If there is no valid LM-move, loose squares can be in one of the four configurations (shown in orange).}
%     \label{fig:loose}
% \end{figure}

% \begin{observation}
%     \label{obs:loose}
%     Let~$C$ be a heavy leaf chunk, and let~$\sigma$ be the set of squares in the boundary cycle of~$C$. If~$C$ admits no valid LM-moves, its loose squares can only be in four configurations (see Figure~\ref{fig:loose}): 
%     \begin{enumerate}
%         \item with $x$-coordinate~$0$ and a square $b \in \sigma$ as its \textsc{n-} or \textsc{s}-neighbor;
%         \item with $y$-coordinate~$0$ and a square $b \in \sigma$ as its \textsc{e} or \textsc{w}-neighbor;
%         \item with squares $b_1, b_2 \in \sigma$ as its \textsc{e}- and \textsc{ne}-neighbors, and no square from~$\sigma$ as its \textsc{se}-neighbor (the \textsc{se}-neighbor may be a loose square);
%         \item with squares $b_1, b_2 \in \sigma$ as its \textsc{n}- and \textsc{ne}-neighbors, and a loose square as its \textsc{nw}-neighbor.
%     \end{enumerate}
% \end{observation}

\subparagraph*{Corner moves.}
LM-moves on their own are not necessarily sufficient to compact a chunk into a suitable left-aligned shape. For example, consider the configuration in Figure~\ref{fig:mirrored-gamma}a, which does not admit any valid LM-moves. However, it has a concave corner that we can fill with two moves (see Figure~\ref{fig:mirrored-gamma}b), to expand the chunk in that direction. Repeating such corner moves allows us to make the chunk in the example left-aligned.

\begin{figure}[b]
    \centering
    \includegraphics{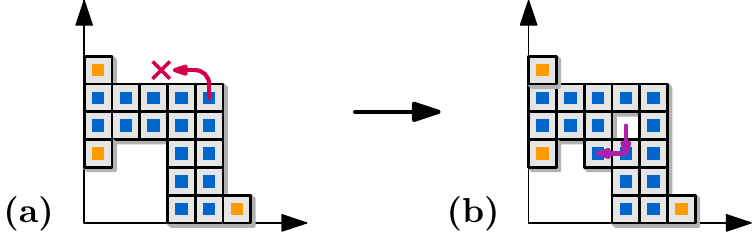}
    \caption{\textbf{(a)} A configuration that does not admit LM-moves. For example, an \dirNW-move (in red) of the top-right square is not valid. \textbf{(b)} Two slide moves expand the concave corner in \dirSW direction.}
    \label{fig:mirrored-gamma}
\end{figure}

We define corners of a chunk~$C$ with boundary cycle~$\sigma$ as follows. A \emph{top corner} (Figure~\ref{fig:corner-moves}a--d) is an empty cell with squares $b_1, b_2, b_3 \in \sigma$ as \dirN-, \dirNE-, and \dirE-neighbors. Similarly, a \emph{bottom corner} (Figure~\ref{fig:corner-moves}e--h) is an empty cell with squares $b_1, b_2, b_3 \in \sigma$ as \dirS-, \dirSE-, and \dirE-neighbors. Note that a corner can be either inside a hole in~$C$ (\emph{internal corner}), or on the outside of~$C$ (\emph{external corner}), and we treat both of these in the same way.

\begin{figure}[h]
    \centering
    \includegraphics{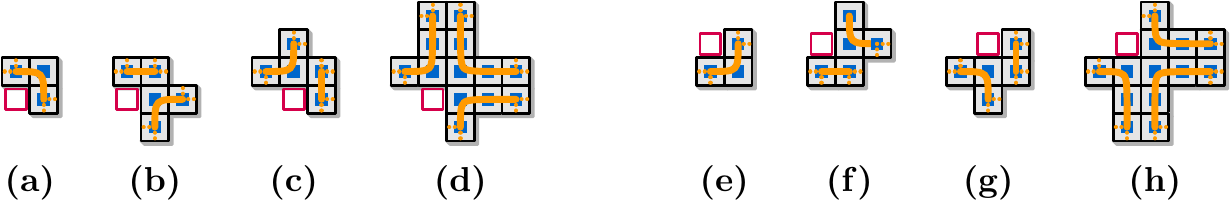}
    \caption{Empty squares shown in red are corners. The boundary cycle of the chunk is shown in orange. \textbf{(a)--(d)} Top corners; \textbf{(e)--(h)} bottom corners. Corners shown in \textbf{(a)} and \textbf{(e)} can be both external and internal. Corners shown in \textbf{(b)}--\textbf{(d)} and \textbf{(f)}--\textbf{(h)} can only be internal.}
    \label{fig:corner-moves}
\end{figure}
Let~$s$ be a top corner in~$C$ with neighbors $b_1, b_2, b_3$ as above. In the case where $b_1, b_2, b_3$ are consecutive squares in $\sigma$ (Figure~\ref{fig:corner-moves}a), we can fill~$s$ by two slide moves: first move either $b_1$ or~$b_3$ into~$s$, and then move~$b_2$ into the cell left empty by the first move. We call this a \emph{top corner move}. We can fill a bottom corner with consecutive $b_1, b_2, b_3$ (Figure~\ref{fig:corner-moves}e) in the same way, just mirrored vertically (a \emph{bottom corner move}). Just like for LM-moves, we say that a corner move is \emph{valid} if all squares $s' \in C$ are still in a single chunk after the corner move. Note that all corners where $b_1, b_2, b_3$ are not consecutive in $\sigma$ (Figure~\ref{fig:corner-moves}b--d and \ref{fig:corner-moves}f--h) do not allow valid corner moves, as $b_1$ and/or $b_2$ becomes a connector. 

\subparagraph*{Chain moves.}
Besides LM- and corner moves, we need a special move to prevent getting stuck when each LM-move is invalid because it would move outside the bounding box~$B$. For example, some squares on the bottom row or leftmost column of~$B$ would be able to perform LM-moves if they were situated in any other row/column of~$B$, as shown in Figure~\ref{fig:chain-move}.

\begin{figure}[h]
    \centering
    \includegraphics{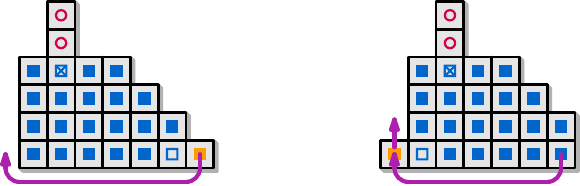}
    \caption{Examples of chain moves: a square moves outside $B$ to the first empty cell closer to the origin. This may require a loose square to move as well.}
    \label{fig:chain-move}
\end{figure}

A \emph{chain move} is a series of moves that is started by such an LM-move that violates validity only by leaving bounding box~$B$. A chain move for a square~$s$ in the bottom row of~$B$ requires an empty cell $e = (x, 0)$ closer to the origin, and works as follows. Square~$s$ must be able to perform an LM-move, more precisely an \dirSW-move that is invalid only because it leaves~$B$. We want to place~$s$ in this empty cell~$e$, unless it creates a link component, which happens only if the square on position $(x+1,0)$ is a loose square. We slide such a loose square upwards with a \dirN-move, and identify the emptied cell as~$e$. Note that $e$ is again the closest empty cell in the bottom row, closer to the origin. We can then move~$s$ to~$e$ by performing an \dirSW-move, a series of \dirW-moves, and finally a \dirWN-move into~$e$. For a square~$s$ in the leftmost column, the direction of all moves is mirrored in $x=y$.

\begin{figure}[t]
    \centering
    \includegraphics{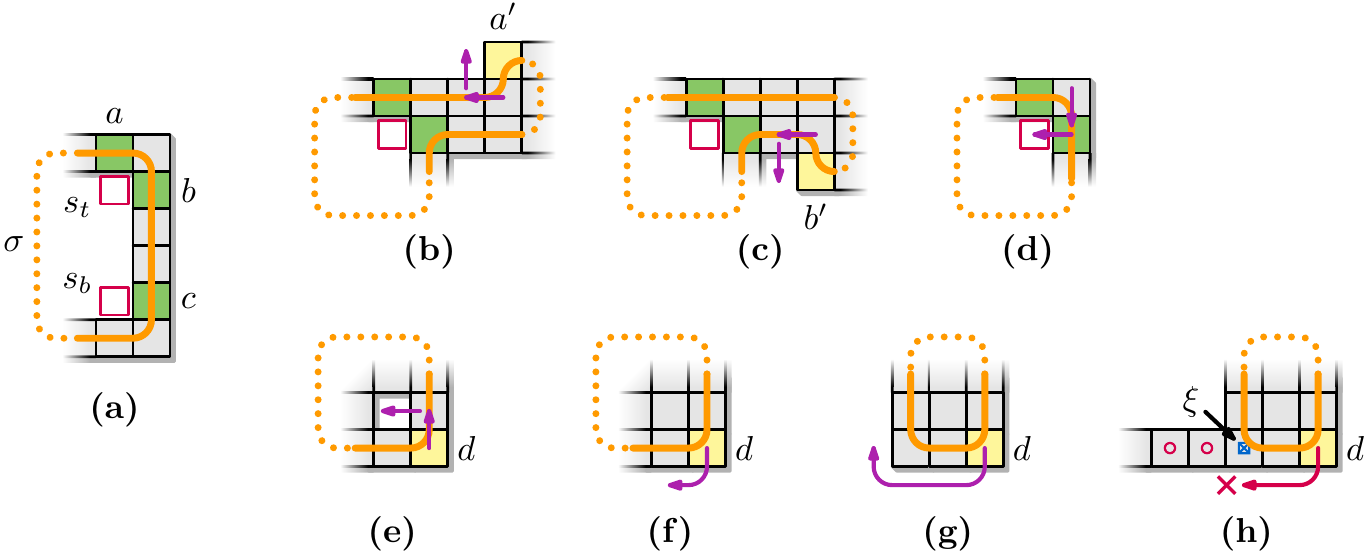}
    \caption{Lemma~\ref{lem:histogram}: \textbf{(a)} A chunk hole with $s_t$ and~$s_b$ marked. \textbf{(b)}--\textbf{(d)} Possible valid moves for~$\sigma'$.  \textbf{(e)}--\textbf{(g)}
    Possible valid moves for~$\sigma''$.
    \textbf{(h)} If a chain move is not possible then $\xi$ is left of $d$.} \label{fig:corner-moves-type-b}
\end{figure}

\begin{restatable}{lemma}{histogram}
    \label{lem:histogram}
    Let~$C$ be a leaf chunk that does not admit valid LM-moves, corner moves, or chain moves. Then $C$ is solid, and its boundary cycle $\sigma$ outlines a left-aligned histogram.
\end{restatable}
\begin{proof}[Proof sketch]
    To show that $C$ is solid, assume to the contrary that $C$ has a hole. Consider the top- and bottommost empty squares $s_t$ and $s_b$ of the rightmost column of empty squares of any hole in~$C$ ($s_t$ may be equal to~$s_b$). Let $a$ and~$b$ be the \dirN- and \dirE-neighbors of~$s_t$, respectively, and let $c$ be the \dirE-neighbor of~$s_b$ (see Figure~\ref{fig:corner-moves-type-b}a). We know that $a, b, c \in \sigma$, because otherwise moving $a$ or~$b$ into~$s_t$, or $c$ into~$s_b$, would be a valid LM-move. $C$ has at most one connector~$\xi$, which is part of~$\sigma$. We can now show, using a case distinction, that $\xi$ lies strictly between $a$ and~$b$ on~$\sigma$ (see Figure~\ref{fig:corner-moves-type-b}b--d), and also that $\xi$ lies strictly between $c$ and~$a$ (see Figure~\ref{fig:corner-moves-type-b}e--h). As these parts of~$\sigma$ are disjoint, they cannot both contain~$\xi$, which results in a contradiction. Hence, $C$ is solid. To show that $\sigma$ outlines a left-aligned histogram, we observe that any external corner admits a valid corner move. Therefore, such external corners cannot exist and thus $\sigma$ outlines a left-aligned histogram (see Figure~\ref{fig:hist-chunk}).
\end{proof}

\begin{figure}
    \begin{minipage}[t]{0.48\textwidth}
        \centering
        \includegraphics{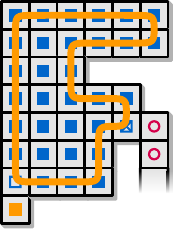}
        \caption{The boundary cycle~$\sigma$ of the leaf chunk $C$ outlines a left-aligned histogram.}
        \label{fig:hist-chunk}
    \end{minipage}
    \hfill
    \begin{minipage}[t]{0.48\textwidth}
        \centering
        \includegraphics{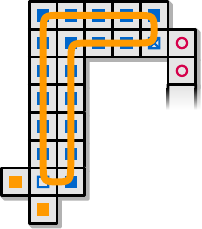}
        \caption{The chunk~$C$ forms a double-$\Gamma$ shape, with one or two loose squares (orange).}
        \label{fig:L-chunk}
    \end{minipage}
\end{figure}

\begin{restatable}{lemma}{doubleGamma}
    \label{lem:gamma}
    Let the component tree~$T$ have more than one leaf, and let~$C$ be a leaf chunk that does not contain the origin and does not admit valid LM-moves, corner moves, or chain moves. Then the boundary cycle of~$C$ outlines a double-$\Gamma$ (see Figure~\ref{fig:L-chunk}).
\end{restatable}

% \begin{proof}[Proof sketch]
% Let $Q$ be the set of loose squares in $C$ and denote $C\setminus Q$ as $C^*$.
% By Lemma~\ref{lem:histogram}, $C^*$ is a left-aligned histogram.
% %Note that $C$ cannot contain a loose square in its topmost row.
% %Otherwise, such a loose square would have a valid \textsc{w}-, \textsc{ws}-move, or a vertical chain move if it has an \textsc{s}-neighbor in $\sigma$, or it would have an \textsc{s}-move if it has a \textsc{w}- or \textsc{e}-neighbor in $\sigma$.

% We consider the rows $\{r_1,r_2,\dots\}$ of $C^*$ from top to bottom.
% We repeatedly apply the argument below to show that, if there are no available moves then (1) the connector of $C$ is the rightmost square in either $r_1$ or $r_2$, and (2) every row $r_i$ with $i>2$ has exactly two squares.
% From this the lemma follows.

% The argument we use is the following. Consider rows $\{r_i,r_{i+1},\dots\}$ of $C^*$ from top to bottom starting with some $i$. Since $C^*$ is a left-aligned histogram, all rows must be left aligned. Consider the first row $r_k$ in the sequence such that $|r_{k-1}|\le|r_k|$ and $|r_k|>|r_{k+1}|$, where $r_{k+1}$ can be an empty row if $r_k$ is the bottom row of $C^*$. Either the rightmost square in $r_k$ is a connector, or it has a valid \textsc{sw}-move, or it has a horizontal chain move (if $r_k$ is the bottom row of $C^*$). Both latter cases lead to a contradiction.
% \end{proof}

\begin{lemma}\label{lem:xy-monotone}
Configuration $\mathcal{C}$ stays in its original bounding box and is $xy$-monotone, once compaction is completed and no valid LM-, corner, and chain moves are possible.
\end{lemma}
\begin{proof}
In the compaction phase we iteratively apply LM-, corner, and chain moves on the configuration in which each leaf chunk contains at least $P$ squares. After compaction, a leaf chunk can either contain the origin, or not. Consider such a chunk $C$ that does not contain the origin. By Lemma~\ref{lem:gamma}, the cycle of $C$ will outline a double-$\Gamma$, and since $C$ is a leaf, it has at least $P$ squares. During gathering and compacting, squares always move to a cell inside the initial bounding box~$B$ of the configuration. Even if a square moves outside~$B$ during gathering or chain moves it always ends inside~$B$ (by Lemma~\ref{lem:light-square} and by definition, respectively). Hence the connector of $C$ will also be inside $B$. Since $P$ is the perimeter of~$B$, any double-$\Gamma$ of $P$ squares, completely inside $B$, will reach the bottom left corner of $B$. Thus, $C$ must contain the origin, as one of the top two rows connects to the connector inside~$B$. As a result, every leaf chunk contains the origin at some point during compaction.

Once every leaf chunk contains the origin, the whole configuration is one single chunk: all leaves of the component tree form a single component now. Continuing the compacting hence results in a left-aligned histogram, by Lemma~\ref{lem:histogram}. Finally consider the topmost row $r$ of this histogram that is longer than the row below it. During compaction, the rightmost cube of~$r$ can perform a valid LM-move, namely a \dirS- or \dirSW-move to the row below it. Note that these moves cannot put cubes outside of the original bounding box~$B$ of~$\mathcal{C}$. Thus, once the compacting phase is completed, the configuration is $xy$-monotone inside $B$.
\end{proof}

\begin{lemma}\label{lem:compaction}
    There is a sequence of $O(Pn)$ moves which reconfigures a configuration in which all leaves are chunks of size at least $P$ to an $xy$-monotone configuration.
\end{lemma}
\begin{proof}
    Using Lemmata~\ref{lem:histogram},~\ref{lem:gamma}, and~\ref{lem:xy-monotone}, we can transform a configuration~$\mathcal{C}$, in which all leaves are chunks of size at least $P$, to an $xy$-monotone configuration. Let $s = (x, y)$ be a square in~$\mathcal{C}$. We assign to~$s$ the score $d(s) = 2x + y$, and let $d = \sum_{s \in \mathcal{C}} d(s)$. Each LM-, and bottom corner move performed in~$\mathcal{C}$ decreases~$d$ by at least~$1$, while every top corner move decreases~$d$ by~two. Initially, $d \leq |\mathcal{C}| \cdot P$, so the total number of LM- and corner moves is also at most $|\mathcal{C}| \cdot P$. Every square~$s$ is involved in at most $P/2$ chain moves, since each chain move places $s$ closer to the origin in the bottom row/leftmost column. Furthermore, every chain move adds at most one additional move for a loose square, which increases the above score by at most two, hence the total number of moves as a result of chain moves is also at most~$O(|\mathcal{C}| \cdot P)$.
\end{proof}

\section{Experiments}\label{sec:experiments}

We experimentally compared our {Gather\&Compact} algorithm to the JavaScript implementation\footnote{\url{https://dccg.upc.edu/people/vera/TFM-TFG/Flooding/}\label{foot:MSlink}} 
of the in-place modification by Moreno and Sacrist\'an~\cite{BachelorMoreno,flooding}
of the Dumitrescu and Pach~\cite{dumitrescu-pushing-squares-2006} algorithm, which we refer to as {MSDP} in the remainder of this section. The original algorithm by Dumitrescu and Pach always requires $\Theta(n^2)$ moves, since it builds a horizontal line to the right of a rightmost square as canonical shape. The in-place modification of Moreno and Sacrist\'an has the potential to be more efficient in practice, since it builds a rectangle within the bounding box of the input. 

\begin{table}[b]
    \centering
    \small
    \begin{tabu}{X[0.33,c] X[r]X[r]X[0.9,r] X[1.16,r]X[1,r]X[r]}
        \toprule
         & \multicolumn{3}{c}{{\bfseries Gather\&Compact}} & \multicolumn{3}{c}{\textbf{MSDP}} \\
        \rowfont[c]{\bfseries}
        $D$ & $50\%$ & $70\%$ & $85\%$ & $50\%$ & $70\%$ & $85\%$ \\
        \cmidrule(lr){1-1}\cmidrule(lr){2-4}\cmidrule(lr){5-7}
        $10$ & 237 \textcolor{gray}{$31\%$} & 156 \textcolor{gray}{$16\%$} & 95 \textcolor{gray}{$8\%$} & 502 \textcolor{gray}{$19\%$} & 427 \textcolor{gray}{$21\%$} & 233 \textcolor{gray}{$35\%$} \\
        $32$ & 5.395 \textcolor{gray}{$\phantom{3}4\%$} & 4.188 \textcolor{gray}{$\phantom{1}5\%$} & 2.529 \textcolor{gray}{$8\%$} & 28.759 \textcolor{gray}{$12\%$} & 18.447 \textcolor{gray}{$13\%$} & 10.027 \textcolor{gray}{$\phantom{3}8\%$} \\
        $55$ & 25.916 \textcolor{gray}{$\phantom{3}2\%$} & 20.024 \textcolor{gray}{$\phantom{1}3\%$} & 12.124 \textcolor{gray}{$4\%$} & 193.390 \textcolor{gray}{$\phantom{1}8\%$} & 116.431 \textcolor{gray}{$12\%$} & 61.617 \textcolor{gray}{$\phantom{3}8\%$} \\
        $80$ & 77.745 \textcolor{gray}{$\phantom{3}2\%$} & 60.516 \textcolor{gray}{$\phantom{1}2\%$} & 36.395 \textcolor{gray}{$3\%$} & 638.847 \textcolor{gray}{$12\%$} & 344.529 \textcolor{gray}{$\phantom{1}9\%$} & 235.413 \textcolor{gray}{$\phantom{3}5\%$} \\
        $100$ & 150.666 \textcolor{gray}{$\phantom{3}1\%$} & 118.232 \textcolor{gray}{$\phantom{1}2\%$} & 69.488 \textcolor{gray}{$3\%$} & 1.318.232 \textcolor{gray}{$11\%$} & 743.133 \textcolor{gray}{$17\%$} & 513.113 \textcolor{gray}{$\phantom{3}7\%$}\\
        \bottomrule\smallskip
    \end{tabu}
    \caption{The number of moves for Gather\&Compact and MSDP on various grid sizes ($D \times D$, such that $P = 4D$) and densities (in \% of $D \times D$). Averages and standard deviations (in \% of average) over 10 randomly generated instances are shown.}
    \label{tab:experiments}
    \vspace{-1.5em}
\end{table}
We captured the output (sequence of moves) of {MSDP} and reran the reconfiguration sequences in our tool, to be able to verify movement sequences, count moves, and generate figures. Doing so, we discovered that MDSP was occasionally executing illegal moves, see Appendix~\ref{sec:bugs} for details and for our corresponding adaptations. Some of these issues could be traced to the same origin: MSDP is breaking convex transitions into two separate moves and sometimes acts on the illegal intermediate state. The number of moves we report in Table~\ref{tab:experiments} counts one move both for convex transitions and for slides; hence the numbers can be lower than the numbers Moreno and Sacrist\'an report. However, our adaptations do replace illegal moves with the corresponding (and generally longer) legal movement sequences, and hence the number of moves can also be higher than those they report.

\begin{figure}[t]
    \centering
    \includegraphics{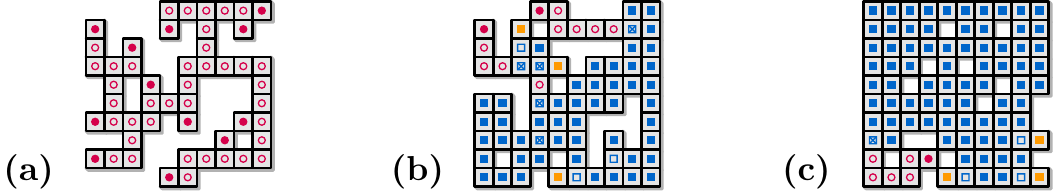}
    \caption{Example input instances on a $10 \times 10$ grid: density \textbf{(a)} $50\,\%$; \textbf{(b)} $70\,\%$; \textbf{(c)} $85\,\%$.}
    \label{fig:inputs}
\end{figure}

\begin{figure}[b]
    \centering
    \includegraphics{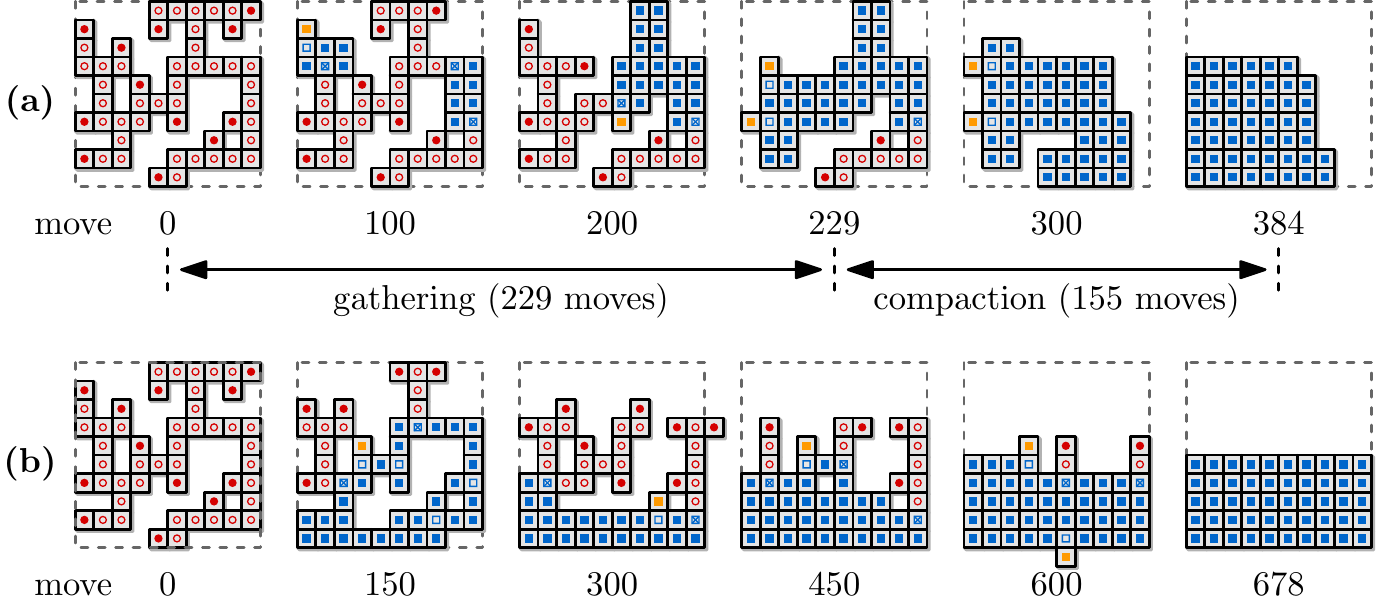}
    \caption{Execution of the two algorithms on one of the input instances for grid size $10 \times 10$, density $50\,\%$. \textbf{(a)} Gather\&Compact; \textbf{(b)} MSDP. Video: \includegraphics[width=2.5mm]{video.pdf} \url{https://tinyurl.com/alga10x10}}
    \label{fig:10x10}
\end{figure}

We use square grids of sizes $10 \times 10$, $32 \times 32$, $55 \times 55$, $80 \times 80$, $100 \times 100$ for our experiments. The data sets for {MSDP} were created by hand and are not available.\footnote{V. Sacrist{\'a}n, personal communication, April 2021.} We attempted to create meaningful data sets of the same nature by starting with a fully filled square grid and then removing varying percentages of squares while keeping the configuration connected. We arrived at three densities, namely ($50\,\%$, $70\,\%$, $85\,\%$), which arguably capture the different types of inputs well (see Figure~\ref{fig:inputs}). 
For each value, the density of the configurations generated is close to homogeneous. 
The configurations with $85 \%$ density are a generalization of the ``dense'' configurations in the data sets for {MSDP}. 
The configurations with $70 \%$ density correspond to the so-called ``medium'' configurations in the data sets for {MSDP}, which combine the two different substructures considered for that density. 
The edge-adjacency graphs of the configurations with $50 \%$ density are essentially trees and, especially in the larger configurations, many leaves are not on the outer boundary (resembling the ``nested'' configurations in the evaluation of {MSDP}). 
For both algorithms we count moves until they reach their respective canonical configurations.
Our online material\footnote{\url{https://alga.win.tue.nl/software/compacting-squares/}} contains our code for {Gather\&Compact}, the input instances, and the adapted version of MSDP.

% : a rectangle with a partial top row on the bottom of the bounding box for MSDP and a solid $xy$-monotone configuration within the bounding box for {Gather\&Compact}.

% The compaction step of {Gather\&Compact} does not rely on any particular order of the available valid moves. Our implementation prioritizes squares by descending $L_\infty$-distance to the origin. We also prioritize downwards LM-moves (\dirW, \dirWS, \dirSW, \dirS) over upwards LM-moves (\dirWN, \dirNW), and top corner moves over bottom moves.

Table~\ref{tab:experiments} summarizes our results and Figure~\ref{fig:10x10} shows snapshots for both algorithms on a particular instance. We observe that Gather\&Compact always uses significantly fewer moves than MSDP, even on high density instances where most squares are already in place. This is likely due to the fact that MSDP walks squares along the boundary of the configuration, while Gather\&Compact shifts squares locally into better position. Figure~\ref{fig:10x10}b shows this behavior at move 600 where one can observe a square on its way along the bottom boundary.

\section{Conclusion}

We introduced the first universal in-place input-sensitive algorithm to solve the reconfiguration problem for the sliding cube model in two dimensions. Our Gather\&Compact algorithm is input-sensitive with respect to the size of the bounding box of the source and target configurations. We experimentally establish that Gather\&Compact not only improves the existing theoretical bounds, but that it also leads to significantly fewer moves in practice.

We showed that minimizing the number of sliding moves for reconfiguration is NP-complete in two dimensions.
The question then arises whether the problem admits approximation algorithms.
% In particular, is there a polynomial-time constant-factor approximation algorithm? 
% In order to prove approximation factors, one must obtain appropriate lower bounds on the number of necessary moves for any given instance. This has remained elusive to the authors.
Our NP-hardness proof can be adapted to show APX-hardness in the 3D sliding cube model 
and we conjecture that the problem is also APX-hard for sliding squares.
%Furthermore, we conjecture that our NP-hardness proof can be adapted to also show APX-hardness.

% However, there may still be room to improve on the bound in this paper. Specifically, it may be possible to improve the hidden constants, by gathering to leaf chucks of size $P/2$ instead of $P$. These chunks still have enough squares to reach the origin, but they have to give up 2-connectivity to do so, and hence the algorithm becomes more complex. Once all leaves create an $xy$-monotone path to the origin, the configuration again consists of a single chunk, and thus the remaining parts of our algorithm still apply.

% Finally, extending our algorithm to three dimensions is currently work in progress. While well-connected components can also be transformed more efficiently in 3D, the algorithm may require higher degrees of connectivity than 2-connectivity.

\bibliography{ref}

\newpage
\appendix

\section{NP-hardness proof}
\label{app:NPhard-sq}

We provide a gadget-based reduction from \textsc{Planar Monotone Rectilinear 3SAT} to prove the following result:

\NPhard*

\paragraph*{Variable gadget}

\begin{figure}[b]
    \centering
    \includegraphics[width=\textwidth,page=2]{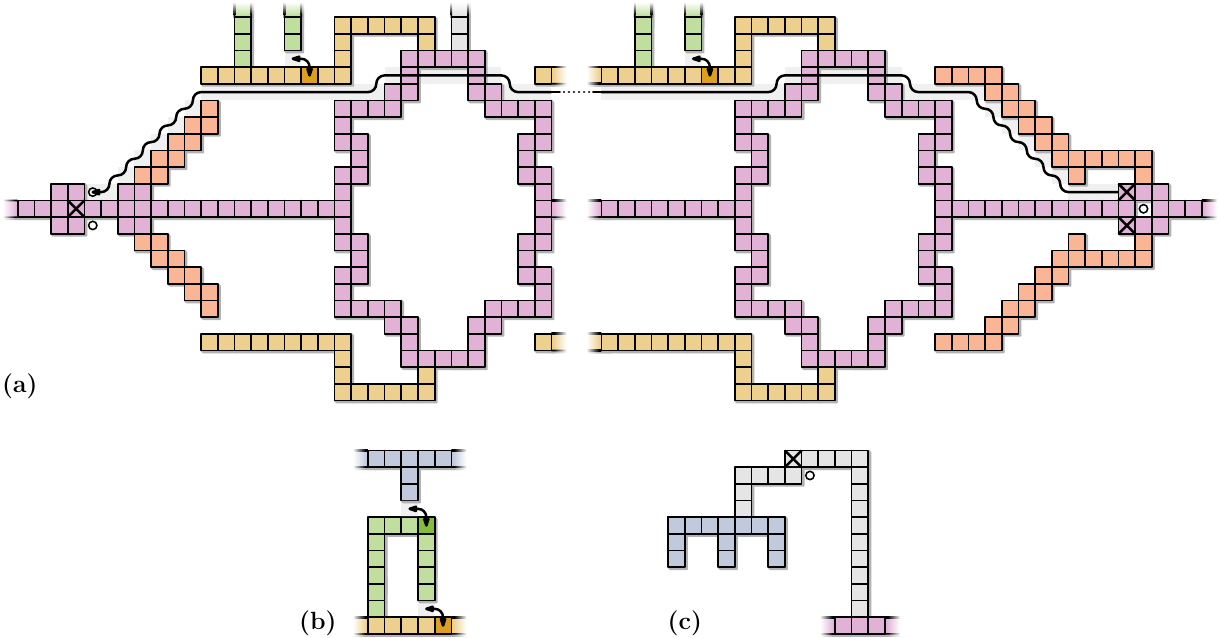}
    \caption{\textbf{(a)} Variable gadget. \textbf{(b)} Wire gadget. \textbf{(c)} Clause gadget.}
    \label{fig:gadgets-app}
\end{figure}

A SAT formula in conjunctive normal form (CNF) is said to be {\em monotone} if its clauses contain only positive literals or only negative literals. Similarly, a CNF formula is {\em planar} if the variable-clause incidence graph is a planar graph. 
It is NP-hard to determine whether a CNF formula $\mathcal{C}$ in which every clause contains three literals (some maybe duplicated) is satisfiable, even if we know that $\mathcal{C}$ is planar and monotone~\cite{monotone3sat}. 

It is known that the incidence graph of a planar monotone CNFs can be drawn in a {\em rectilinear} way such that $(i)$ all variables are on the $y=0$ axis, $(ii)$ clauses with positive variables are on the $y>0$ halfplane (or $y<0$ halfplane otherwise), and $(iii)$ edges that connect variables with the clauses they appear in can be drawn with rectilinear edges without crossings (see Figure~\ref{fig:monotone3SAT}).
%\todo{Put a figure: ideally we want a planar drawing of the cnf and our reduction side by side.}

All variables appear sequentially on the $x$-axis in a rectilinear planar drawing of a planar monotone 3SAT instance. 
For each variable, we create the gadget shown in Figure~\ref{fig:gadgets-app}a. 
The variable gadgets are connected by a horizontal line of squares; the spacing between the gadgets is described afterwards but we can assume that they are far enough apart that they cannot interact with each other. 
With these connections our reduction has a central path of cycles that goes through all variable gadgets (depicted in pink in Figure~\ref{fig:gadgets-app}a). 
    
We note the following properties of the variable gadget:

\begin{itemize}
    \item The variable gadget is horizontally symmetric. 
    
    \item We mark the start and end of a variable gadget with three small cycles.  
    %(also show in in pink in Figure~\ref{fig:gadgets}a). 
    Two of them form a $2\times 3$ rectangle and are on the left side. 
    The other one fits in a $3\times 3$ grid and is on the right side. 
    Note how additional squares are also attached to the middle and rightmost small cycles, forming staircase-like shapes. 
    These additional squares are shown in orange in Figure~\ref{fig:gadgets-app}a.
    
    \item The gadget associated to variable $x_i$ has $k_i$ \diam-shaped cycles in the path of cycles, where~$k_i$ is the number of times that the variable $x_i$ appears in $\mathcal{C}$.
    These \diam-shaped cycles are of constant size (contained in a rectangle of size $13\times 19$ units). 
    %They are both vertically and horizontally symmetric.
    We space consecutive cycles of the same gadget by a line of $7$ squares (see Figure~\ref{fig:reduction}). 
    The spacing between the rightmost and leftmost \diam-shaped cycles and the closest small cycle that represent the end of the variable gadget is 11 units.
    
    \item Each \diam-shaped cycle has two {\em prongs}. Each prong is a path of 18 squares forming an almost horizontal segment that extends to the left of the cycle (depicted in yellow in Figure~\ref{fig:gadgets-app}a). 
    %The loop and the two prongs are horizontally symmetric. 
    We space consecutive cycles of the same gadget by $7$ squares from each other (see Figure~\ref{fig:reduction}).
    
    \item The source and target configurations are very similar. 
    The only differences are three positions in the leftmost and rightmost cycles that mark the start and end of the gadget. 
    Positions that initially start occupied and must be emptied are marked with $\times$ in Figure~\ref{fig:gadgets-app}a. 
    Similarly, there are three positions that start empty and must be occupied (marked with~$\circ$). 
    Note that the horizontal distance between a $\times$ square on the right side and a $\circ$ position on the left side is $5+11+20k_i-7+11 = 20k_i + 20$. 
\end{itemize}

\varlem*
% \begin{lemma}\label{lem:variable}
% The reconfiguration of variable gadget $x_i$ needs at least $20k_i+24$ many moves. 
% Moreover, the only way to achieve that number moves implies transferring one of the right $\times$ squares to the left $\circ$ position at the same height along the path shown in Figure~\ref{fig:reduction} or along the horizontally symmetric one.
% \end{lemma}
\begin{proof}
First we argue a lower bound by an earth moving argument: a reconfiguration need not precisely move the $\times$ squares to the $\circ$ positions (in fact, as we will see later, this is not the case in an optimal solution). However, any solution must transfer mass from the $\times$ squares to the $\circ$ empty grid cells. Thus, the minimum distance matching between the $\circ$ empty grid cells and the $\times$ squares is a trivial lower bound for the number of moves required.

Recall that the variable gadgets will be spaced so that they cannot interfere with each other. 
The same spacing will be applied vertically ensuring that $\times$ squares and $\circ$ positions in a variable gadget must be transferred to each other. 
Any perfect matching between $\times$ squares and $\circ$ positions in a variable gadget contains a pair consisting of a $\times$ square on the right and a $\circ$ position on the left. 
As mentioned above, the horizontal distance between the $\times$ squares on the right and the $\circ$ positions on the left and is $20k_i+20$. 
The other two pairs in the matching need at least two extra moves; we can see that they are additional to the $20k_i+20$ required moves since they must happen to the left or right of those moves. This means that any solution must use at least $20k_i+24$ moves.

Thus, to match the $20k_i+24$ bound, we need to transfer a $\times$ square on the right to a $\circ$ position on the left using only $x$-monotone moves. 
More precisely, each move should make progress to the left, so for each two consecutive $x$-coordinates between the two positions there can only be exactly one $x$-monotone move in an optimal sequence on moves.
We claim that there are only two such sequences.
This is based on the fact that there are only two possible strictly $x$-monotone ways to traverse a \diam-shaped cycle $\sigma$. 
We argue about the top half of the cycle; the other half is symmetric. 
Note that the sequence of $x$-monotone moves cannot reach the top row of $\sigma$. 
The second row from the top allows optimal traversal, as depicted in Figure~\ref{fig:reduction}. 
It requires moving a square in $\sigma$ in that row to the left when its position can be afterwards filled with a square moving from the right (see Figure~\ref{fig:sequence-var}a). 
To make progress to the left moving a square in any other row of $\sigma$ is not possible in an optimal sequence as the position of the square moved cannot be replaced using a sequence of strictly $x$-monotone moves and maintaining connectivity at all times. 

\begin{figure}[tb]
    \centering
    \includegraphics[width=\textwidth,page=3]{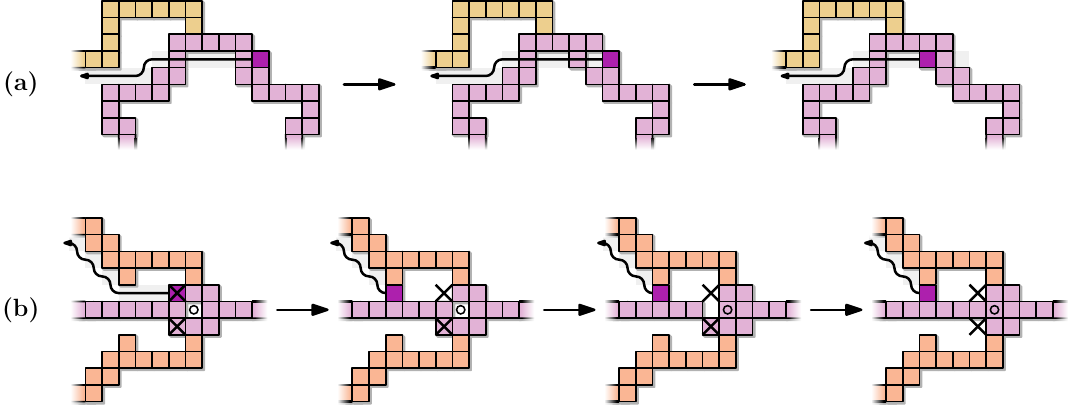}
    \caption{Sequence of moves to reconfigure a variable gadget.}
    \label{fig:sequence-var}
\end{figure}

Thus, there are two valid optimal ways to traverse a \diam-shaped cycle:
%one on the top and one on the bottom. 
either using a horizontal path one unit below the top of the cycle or using a horizontal path one unit above the bottom of the cycle. 
It is only possible to connect two such paths in an $x$-monotone way if they are both on along the top or along the bottom of the cycles, and it is easy to verify that such $x$-monotone connecting sequences are unique.
%Once we have reached that position, two more moves are needed to move the $\times$ square that started in the left with the remaining $\circ$ position in the left side.
We have therefore shown that to transfer a $\times$ square on the right to a $\circ$ position on the left we must use the path in Figure~\ref{fig:reduction} or the horizontally symmetric one. 

In order to reconfigure the whole variable gadget in $20k_i+24$, the moves must be done in a very specific order.  
We might want to start by doing the moves that transfer 
a $\times$ square on the right to the nearby $\circ$ position. 
For that we would need to start by moving a square that is edge-adjacent to both positions, but this is not possible because of the \emph{backbone} requirement (during the move the configuration would not be connected). 
Instead, we start by moving one of the two $\times$ squares on the right to the corresponding $\circ$ empty cell on the left, but we only move three positions along the path. After three movements we create a cycle that contains a square that we initially wanted to move (see Figure~\ref{fig:sequence-var}b). 
Thus, the square originally between both $\times$ squares can slide to the $\circ$ position. 
We then move the $\times$ square in its original position on the right to occupy the just emptied position. After those two moves we can continue moving the $\times$ square that started moving all the way until it reaches the matching $\circ$ position to the right. Finally, two more moves are needed to transfer the left $\times$ square to the remaining $\circ$ position. This is a total of $20k_i+24$ moves as claimed.
\end{proof}

In short, Lemma~\ref{lem:variable} shows that in order to reconfigure from the source to the target configuration using as few moves as possible we must transfer one of the two $\times$ squares on the right to a $\circ$ position on the left. During the move process, the $\times$ square will create cycles involving the central path of cycles and either each of the upper or lower prongs. 
We use these cycles to determine a truth assignment to each variable. 
We assign a variable to \texttt{true} if in the reconfiguration process we create cycles with each of the upper prongs in the gadget associated to that variable. 
We assign it to \texttt{false} if the cycles are created with all the lower prongs instead. 
Note that, if we allow more than $20k_i+24$ moves, we could assign a variable two values (or even to create no cycles and thus have no assignment). 
As we will see later, neither of these cases will be possible in a solution with our specified number of moves.

\paragraph*{Clause and wire gadgets} 

The clause gadget mainly consists of a set of squares forming a \emph{pitchfork} ($\pitchfork$) shape (in blue in Figure~\ref{fig:gadgets-app}c). 
The pitchfork has three \emph{tines} consisting of two squares each. 
Each tine corresponds to a literal in the clause. 
We add a path of squares connecting the $\pitchfork$ shape to the central path of cycles so that the source configuration is connected (drawn in gray).

Similar to the variable gadget, most squares are in both the source and the target configurations. The only exception is one square (marked with $\times$) that wants to be transferred to a nearby position (marked with $\circ)$.  
Note that the reconfiguration could be done with exactly two moves if the connectivity condition was met. 
However, by construction of the gadget, the move is initially not possible as it would disconnect the $\pitchfork$ part. 

The wire gadgets are connected to the variable gadgets and part of them is placed very close to each of the tines of a pitchfork as illustrated in Figure~\ref{fig:gadgets-app}b (the squares associated to the wire gadget are shown in green). 
The goal is to allow creating an additional connection between the $\pitchfork$ of a clause gadget and the variables gadgets in two moves as long as the variable assignment satisfies the clause. 
Once this connection is created the clause gadget can be reconfigured in two moves, as its original connection to the central path of cycles is now part of a cycle. 

The precise placement of the wire gadgets is based on the rectilinear drawing of the planar monotone 3SAT instance. 
Each connection between a clause and a variable is replaced by a wire gadget, which is a path of squares that form a $\sqcap$ shape for positive literals and a $\sqcup$ shape for negative ones. 
More precisely, the path consists of two vertical lines of squares (the left one one unit shorter than the right one) and a horizontal line of four squares. 
% The squares follow the rectilinear path in the drawing until they get very close to one of the endpoints of the variable gadget. After that, the sequence returns back through the same path (two spaces apart). Note that it does not connect a second time to the variable prong, but it is very close as well
Each wire gadget is attached at the left endpoint of the path to a distinct prong of the corresponding variable gadget (recall that each variable gadget contains as many prongs as the number of times that it appears in $\mathcal{C}$, hence it will always be possible to do so). We can use the wires (and their relative order in the embedding) to match each literal in a clause to a prong (note that there will be spare prongs since we made the instance symmetric). 
%Since each wire gadget corresponds to a literal in a clause, we associate the prongs connected to a wire gadget (not all are, but when they are there is only one wire gadget connected to the prong) to that literal. 
%These squares attach to either the upper or lower right prong of the associated variable (depending on whether the literal is positive or negative). 
%The squares follow the rectilinear path in the drawing until they get very close to one of the endpoints of the variable gadget. After that, the sequence returns back through the same path (two spaces apart). Note that it does not connect a second time to the variable prong, but it is very close as well (see Figure~\ref{fig:gadgets-app}b, the squares associated to the wire gadget are shown in green). 
In the wire gadget all squares are in both the source and target configurations. 

The positive clause gadgets are placed so that the bottommost square of a tine is two units above and one unit to the left of the topmost rightmost square of the corresponding wire gadget. 
The connection of the clause gadget to the central path of cycles is done by a rotated L-shaped path that connects the pitchfork to the \diam-shaped cycle to which the wire gadget corresponding to the rightmost literal is attached through a prong. 
The placement of the negative clause gadgets is the horizontally symmetric analogue. 
Note that we need to use several prongs in a variable gadget to guarantee that the wire gadgets can be directly connected to the path of cycles. 
If we were using only one prong for all the wire gadgets, the clause $x_1 \lor x_3 \lor x_5$ in Figure~\ref{fig:monotone3SAT} would not be able to get directly attached (without crossings) to the central path of cycles. 

We place the clause gadgets at different heights increasing the heights and make the vertical separations between gadgets $100m$ units. 
This is the same separation used between variable gadgets. 
This guarantees that, in a reconfiguration that uses
$100m$ moves or fewer, 
%the minimum number of moves, 
any square that comes in contact with the tines must come from either the clause , or one of the wire gadgets that corresponds to a literal in the clause, or a clause connection next to such a wire gadget.

\clauselem*
% \begin{lemma}\label{lem:clause}
% A clause gadget needs at least six moves to be reconfigured and it can be reconfigured with six moves if and only if 
% a prong associated to a literal in the clause is part of a cycle. 
% \end{lemma}
\begin{proof}
%Recall that, by the way that we create a truth assignment from a solution of the square reconfiguration instance, when a literal satisfies the clause we know that the prong that contains the corresponding wire gadget is at some point part of a cycle.
If the condition in the statement is true, at that point we can move one square from the prong so that the wire gadget forms a cycle (see Figure~\ref{fig:reduction}). 
This cycle allows us to move another square of the wire gadget so that 
we create a cycle that goes through the clause gadget. 
Once we have this cycle, we can reconfigure the clause gadget with two moves. 
Finally, we need two more moves to undo the changes created in both the wire gadget and the prong. This is a total of six moves as claimed. 

For the reverse direction, assume that we  
never created a cycle that includes a prong associated to a literal in the clause. 
The separation between gadgets implies that we must spend additional moves to create a cycle that allows to reconfigure the clause gadget 
and the best way to do so is creating (with one additional move) a cycle that includes a prong associated to a literal in the clause. 
This would require strictly more than six moves. 
\end{proof}

\paragraph*{Overall reduction}

%$20k_i+25$
If we put the variable gadgets directly next to each other, their total width is $60m+25n$. 
To make sure that variable gadgets do not interact with each other, 
we make the space between gadgets of variables $x_i$ and $x_{i+1}$ to be $100m+in$ units. 
This spacing fills two roles: similar to the vertical separation between different clause gadgets it prevents interaction between different gadgets. Also, the unique spacing between makes sure that any solution (that uses $O(n+m)$ moves) must match the source and target gadgets (that is, global translations will not help). 
Thus, the whole reduction fits in a rectangle of sidelength $O(n^2+m^2)$ and the construction therefore has polynomial size. 
The next lemma finishes the proof of Theorem~\ref{thm:NP-hard-sq}. 

% \begin{lemma}
% A \textsc{Planar Monotone 3SAT} instance can be solved if and only if the corresponding  reconfiguration problem instance can be solved using $66m+24n$ sliding moves.
% \end{lemma}
\finallem*
\begin{proof}
Assume that the 3SAT instance is satisfiable. 
This means that each variable can be assigned a truth value so that all clauses are satisfied. 
By Lemma~\ref{lem:variable} we know that we can reconfigure the gadget associated to variable $x_i$ in $20k_i+24$ many moves 
using the top path if in the solution $x_i$ is \texttt{true} and the bottom path otherwise. 
We satisfy the variable gadgets one by one. 
When using the top or the bottom path we stop at the positions in which we create a cycle including a prong. 
If this prong is associated to a literal satisfiying a clause that has not been reconfigured yet, we reconfigure this clause as in  Lemma~\ref{lem:clause}. 

By construction the top (bottom) path in a variable gadget closes cycles with all top (bottom) prongs. 
Since we are reducing from rectilinear planar monotone 3SAT and using a truth assignment that satisfies the formula, when the last variable gadget finishes its reconfiguration all clause gadgets have been reconfigured. 
Although we do not have a direct bound on $k_i$, we know that $\sum_i k_i = 3m$ (since each clause contains three variables). 
Thus, it required $\sum_i 20k_i+24 = 60m+24n$ many moves to reconfigure all variable gadgets. 
Each clause was reconfigured using exactly six moves. 
Wire gadgets need no reconfiguration, 
so in total $66m+24n$ many moves were sufficient to reconfigure the whole construction.

Assume on the contrary that the reconfiguration problem instance can be solved using $66m+24n$ sliding moves. 
By Lemma~\ref{lem:variable} we must use at least $20k_i+24$ moves to reconfigure the gadget for variable $x_i$. 
Since the gadgets are independent and $\sum_i 20k_i+24 = 60m+24n$, there are at most $6m$ moves left to also reconfigure all the clause gadgets. 
By Lemma~\ref{lem:clause} we must spend at least six moves to reconfigure each clause gadget. 
It follows from the proof of this Lemma that at least six moves per clause gadget reconfiguration process do not contribute to the reconfiguration of any other gadget. 
Thus, exactly $20k_i+24$ moves are spent to reconfigure the gadget for variable $x_i$ 
and exactly six moves to reconfigure each clause gadget.
By Lemma~\ref{lem:variable}, the prongs that become part of cycles in the reconfiguration of the the gadget for variable $x_i$ are either all the top ones or all the bottom ones. 
Note than when a cycle through a clause gadget is created 
it includes only parts of the prong and wire gadget that contributed to its creation, part of the clause gadget itself, and part of the central path of cycles. 
In particular, no prong becomes part of a cycle as a result of reconfiguring a clause gadget with six moves. 
By Lemma~\ref{lem:clause}, to reconfigure a clause gadget with six moves it must happen that a prong associated to a literal in the clause is part of a cycle.
Since we established that all cycles containing prongs are created in the reconfiguration of variable gadgets, 
the top/bottom path choice in the reconfiguration of variable gadgets defines a truth assignment satisfying the 3SAT formula.
\end{proof}

One final subtlety is that we cannot assume the relative position of the source and target configurations in our construction. 
However, it is easy to force the alignment of all the gadgets (except for the marked positions). 
The current construction has polynomial size. 
Let $N$ be a polynomial function of $n$ and $m$ such that the whole contraction fits in a square grid of sidelength $N$. 
If the instance is satisfiable (and actually otherwise too) and the source and target configurations have the gadgets aligned, we can reconfigure in $O(N)$ moves. 
More precisely, there is a constant $\kappa$ such that we can reconfigure in no more than $\kappa N$ moves. 
We can add to both the source and target configurations 
a line with $\kappa N$ squares 
to the right of the central path of cycles in the rightmost variable gadget. 
To the right endpoint of these lines we attach a set of squares forming a large square shape of sidelength $\kappa ^2 N^2 + N$. 
Note that this only increases the size of the construction by a polynomial factor. 
Any relative placement of the source and target configurations that does not align these large squares (and therefore the gadgets) requires at least $\kappa ^2 N^2$ moves, 
and can therefore not be optimal. 
%place n - n^1/3 modules in a square in a corner of the construction
%scale down the existing construction to use n' = n^1/3
%Translating by 1 costs Omega(n) just for the square, whereas any solution to
%the hardness instance costs O(n^2/3).  By setting the constants right, we can guarantee that there isn't any translation.

\section{Omitted material from Section~\ref{sec:algorithm}}
\label{sec:omitted}

\gatheringLemma*

\begin{proof}
    By Lemma~\ref{lem:light-square}, we can make a light square $s$ part of a chunk in $O(P)$ moves by moving cubes from~$D$, the set of descendants of $s$. This process creates new light squares only if removal of a square~$a$ breaks a cycle in~$D$. Thus, every new light square is part of~$D$.
    
    We repeat the procedure, selecting a light square of maximal capacity at every step. Overall, a square can be light at most once in the process. Thus, after $O(Pn)$ moves no light squares remain, and all the leaves in the component tree are chunks of size at least~$P$. Note that, while the root square always has capacity $n-1$, an adjacent square can have capacity $2 \leq |D| \leq P$, and the resulting chunk will be the root component, either because the root square is on the boundary cycle, or is a loose square. In case the adjacent square is too light, namely $|D| < 2$, then the root component may stay a link.
\end{proof}

\begin{figure}[b]
    \renewcommand{\thefigure}{\ref{fig:corner-moves-type-b}}
    \centering
    \includegraphics{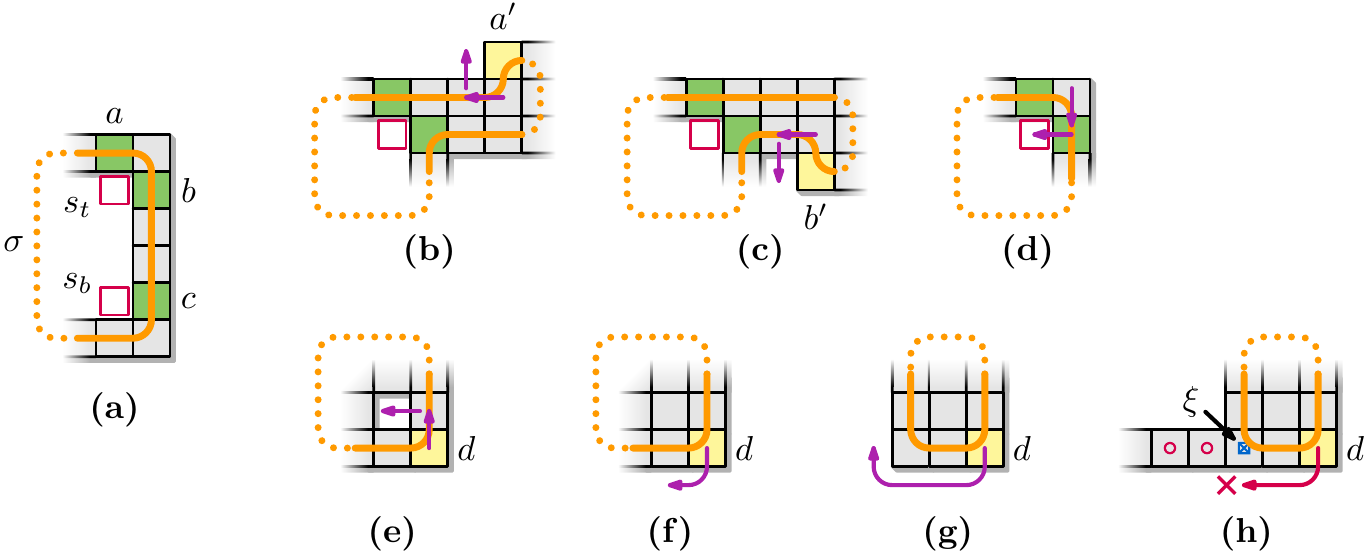}
    \caption{(repeated from main text) Lemma~\ref{lem:histogram}: \textbf{(a)} A chunk hole with $s_t$ and~$s_b$ marked. \textbf{(b)}--\textbf{(d)} Possible valid moves for~$\sigma'$.  \textbf{(e)}--\textbf{(g)}
    Possible valid moves for~$\sigma''$.
    \textbf{(h)} If a chain move is not possible then $\xi$ is left of $d$.}
    \addtocounter{figure}{-1}
\end{figure}

\histogram*
\begin{proof}
    We first show that $C$ is solid. Assume to the contrary that $C$ has a hole. Consider the top- and bottommost empty squares $s_t$ and $s_b$ of the rightmost column of empty squares of any hole in~$C$ ($s_t$ may be equal to~$s_b$). Let $a$ and~$b$ be the \dirN- and \dirE-neighbors of~$s_t$, respectively, and let $c$ be the \dirE-neighbor of~$s_b$ (see Figure~\ref{fig:corner-moves-type-b}a). We know that $a, b, c \in \sigma$, because otherwise moving $a$ or~$b$ into~$s_t$, or $c$ into~$s_b$, would be a valid LM-move. $C$ has at most one connector~$\xi$, which is part of~$\sigma$. We now show that $\xi$ lies strictly between $a$ and~$b$ on~$\sigma$, and also that $\xi$ lies strictly between $c$ and~$a$, to arrive at a contradiction.
    
    Let $\sigma'$ be the part of $\sigma$ strictly between $a$ and~$b$, walking along the boundary of $C$ in the clockwise order. If $\sigma'$ visits the row above~$a$ (leaving the row of~$a$ for the first time at square~$a'$), then there exists a bottom corner move filling the \dirW-neighbor of~$a'$ (see Figure~\ref{fig:corner-moves-type-b}b). This move is valid since the part of $C$ to the right of $s_t$ by definition does not contain any holes. Similarly, if $\sigma'$ visits the row below~$b$ (returning to the row of~$b$ for the last time at square~$b'$), then there exists a top corner move filling the \dirW-neighbor of~$b'$ (see Figure~\ref{fig:corner-moves-type-b}c). We conclude that the part of~$S$ attached to $a$ and~$b$ is a protrusion of two rows tall. The protrusion is non-empty, because otherwise we can either move a loose square with an LM-move or fill~$s_t$ with a corner move (see Figure~\ref{fig:corner-moves-type-b}d). Consider the rightmost column of the protrusion. If this column contains a loose square, we can move it with an LM-move. Hence, the column contains a square in the row of~$a$ and a square in the row of~$b$. If neither of these is~$\xi$, then we can perform an LM-move. Hence, $\sigma'$ contains~$\xi$.
    
    Let~$\sigma''$ be the part of~$\sigma$ strictly between $c$ and~$a$ in clockwise order. Walk over~$\sigma''$ until encountering the first square~$d$ whose \dirN- and \dirW-neighbors are adjacent to~$d$ on~$\sigma$. Assume that $\xi$ is not $d$ or one of its neighbors. If $d$ or the \dirN-neighbor of~$d$ have a loose square attached to them, we can perform an LM-move on this loose square. Otherwise, if the \dirNW-neighbor of~$d$ is part of a hole in~$C$, then this hole can be filled with a bottom corner move (see Figure~\ref{fig:corner-moves-type-b}e). Otherwise, there are three cases: either (1) we can perform an \dirS- or \dirSW-move on~$d$ (see Figure~\ref{fig:corner-moves-type-b}f), or (2) we can perform a horizontal chain move (see Figure~\ref{fig:corner-moves-type-b}g), or (3) the chain move is impossible because no empty square~$e$ is available to move to (see Figure~\ref{fig:corner-moves-type-b}h). In cases (1) and (2), by contradiction $\xi$ is $d$ or one of its neighbors; in case (3), $\xi$ needs to be in the bottommost row. In any case, $\sigma''$ contains~$\xi$.
    
    As $\sigma'$ and~$\sigma''$ are disjoint, they cannot both contain~$\xi$, which results in a contradiction. Hence, $C$ is solid. To show that $\sigma$ outlines a left-aligned histogram, we observe that any external corner (with neighbors $b_1, b_2, b_3$ as defined above) admits a valid corner move. Indeed, none of $b_1, b_2, b_3$ can be a loose square, as those would admit LM-moves. Furthermore, as $C$ is solid, a boundary square cannot be a cut square for a component on the inside of~$\sigma$. Finally, since $C$ is a leaf chunk, the only cut square in~$\sigma$ for a component on the outside is its connector. Since only one out of $b_1$ and~$b_3$ can be this connector, we can perform a valid corner move starting with the other (non-connector) square. Therefore, by our assumption that there are no corner moves in~$C$, there cannot be external corners, and thus $\sigma$ outlines a left-aligned histogram (see Figure~\ref{fig:hist-chunk}).
\end{proof}

\begin{figure}
    \begin{minipage}[t]{0.48\textwidth}
        \renewcommand{\thefigure}{\ref{fig:hist-chunk}}
        \centering
        \includegraphics{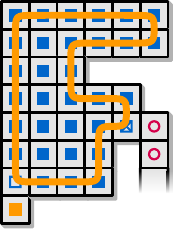}
        \caption{(repeated from main text) The boundary cycle~$\sigma$ of the leaf chunk $C$ outlines a left-aligned histogram.}
        \addtocounter{figure}{-1}
    \end{minipage}
    \hfill
    \begin{minipage}[t]{0.48\textwidth}
        \renewcommand{\thefigure}{\ref{fig:L-chunk}}
        \centering
        \includegraphics{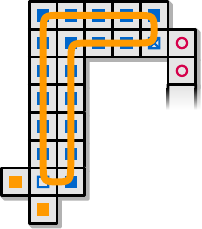}
        \caption{(repeated from main text) The chunk~$C$ forms a double-$\Gamma$ shape, with one or two loose squares (orange).}
    \end{minipage}
\end{figure}

\doubleGamma*
\begin{proof}
    Let~$C^*$ be the set of squares outlined by~$\sigma$. By Lemma~\ref{lem:histogram}, $C^*$ is a left-aligned histogram. Let $\{ r_1, r_2, \ldots \}$ be the rows in~$C^*$, ordered from top to bottom. The connector~$\xi$ of~$C$ lies on~$\sigma$, and thus in~$C^*$. Assume that $\xi$ is in row~$r_i$ ($i \geq 3$). In that case, the leftmost square of~$r_1$ has a valid move~$m$: a \dirWS-move or a vertical chain move. In particular, $m$ is not blocked by a loose square in the leftmost column of~$C$, because that column can contain only one loose square in the last row of~$C^*$ (otherwise one of the other loose squares would admit a valid \dirS-move). Similarly, $C$ cannot contain a loose square, which could block~$m$, in its topmost row. Indeed, such a loose square would admit a \dirW-, \dirWS-move, or a vertical chain move (if it has an \dirS-neighbor in~$\sigma$), or it would admit an \dirS-move (if it has a \dirW- or \dirE-neighbor in~$\sigma$). The existence of~$m$ leads to a contradiction, so $\xi$ lies on $r_1$ or~$r_2$. Because $C^*$ is 2-connected, this implies that $|r_1| = |r_2|$, forming the horizontal leg of the double-$\Gamma$.

%, $|r_{k - 1}| \leq |r_k|$ and $|r_k| > |r_{k + 1}|$, where $|r_{k + 1}| = 0$ if $r_k$ is the bottom row of~$C^*$
    Consider the last row~$r_k$ ($k \geq 3$) such that $|r_k| > 2$. The rightmost square in~$r_k$ has a valid \dirSW-, \dirS-move, or horizontal chain move (which, by a similar argument as before, cannot be blocked by a loose square). On the other hand, $|r_k| \geq 2$, again because of 2-connectivity. Therefore, for all rows $r_i$ ($i \geq 3$), $|r_i| = 2$, forming the vertical leg of the double-$\Gamma$.
\end{proof}

\subsection{Light configurations}\label{sec:light}

We say that a configuration~$\mathcal{C}$ is \emph{light}, if it consists of fewer than $P$ squares, where $P$ is the perimeter of the bounding box of $\mathcal{C}$. Our algorithm, as explained in the main text, cannot directly handle such configurations if $\mathcal{C}$ does not contain the origin: there are too few squares to guarantee that compacting will always result in a chunk that contains the origin. However, we can use a simple preprocessing step to ensure that $\mathcal{C}$ will contain the origin.

For a light configuration~$\mathcal{C}$ which does not contain the origin, we select a stable square as in the gathering phase: a stable square in a link, or an extremal stable square in a chunk. We iteratively move this stable square along the boundary of~$\mathcal{C}$ to the empty cell~$e$ that is the \dirW-neighbor of the root square. We iteratively continue to do so until $\mathcal{C}$ contains the origin. Note that $e$ must necessarily be empty.

At this point, we can simply gather and compact $\mathcal{C}$ and arrive at an $xy$-monotone configuration, for the following reason. The gathering phase works as in the main text, since we can iteratively apply Lemma~\ref{lem:light-square} on the light square closest to the root (which can be the root itself), to get a single chunk. As the root square is located at the origin, and we never move the root during gathering, we get a chunk containing the origin. Similarly, in the compaction phase, this chunk will become a solid left-aligned histogram by Lemma~\ref{lem:histogram}. Since it already contains the origin, and we do only monotone moves towards the origin, this chunk still contains the origin. Finally the topmost row $r$ of this histogram, that is longer than the row below it, still has valid LM-moves. These moves will happen during compacting, hence the result is an $xy$-monotone configuration.

There are at most $P/2$ empty cells to the left of the root. We fill each of these cells by by walking along the boundary of $\mathcal{C}$.
Since configuration~$\mathcal{C}$ consists of less than $P$ squares, this requires at most $O(Pn) = O(P^2)$ moves. Both gathering and compacting take $O(Pn)$ moves, as proven in the main text, so including the preprocessing, we still arrive at a bound of $O(Pn)$ for the number of moves.

\subsection[Transforming xy-monotone configurations]{Transforming $\boldsymbol{xy}$-monotone configurations}
\label{sec:canonical}

After gathering and compacting we arrive at an $xy$-monotone configuration. 
However, this configuration is not unique and hence we need to be able to transform between such configurations. We use a potential function to guide this transformation.

\begin{figure}[t]
    \centering
    \includegraphics{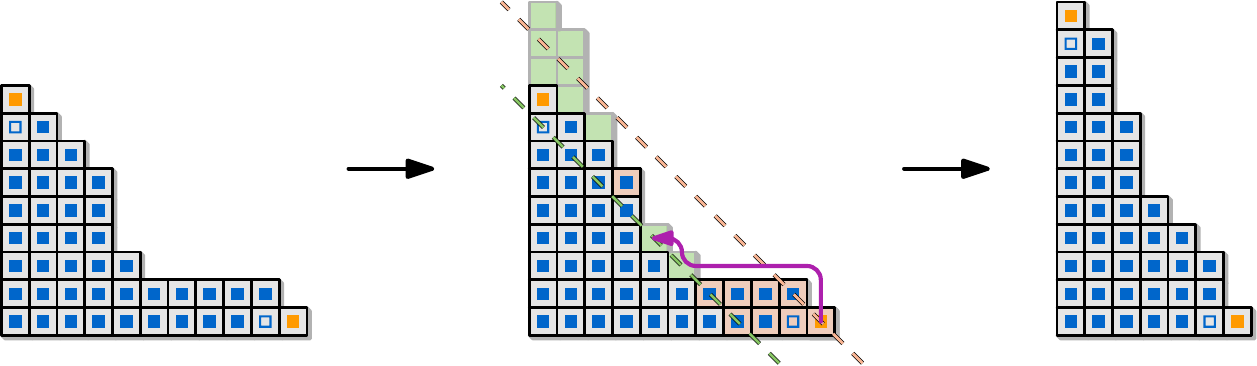}
    \caption{Transforming between two $xy$-monotone configurations. The dashed lines go through cells with the same potential.}
    \label{fig:canonical}
\end{figure}

\begin{lemma}
    Let $\mathcal{C}_1$ and $\mathcal{C}_2$ be two $xy$-monotone configurations of $n$ squares each, let $P$ and $P'$ denote the perimeters of their respective bounding boxes, and let $\bar{P} = \max \{ P, P' \}$. We can transform $\mathcal{C}_1$ into $\mathcal{C}_2$ using at most $O(\bar{P}n)$ moves while maintaining edge-connectivity. 
\end{lemma}
\begin{proof}
    Let $\mathcal{C} \coloneqq \mathcal{C}_1$. For each grid cell $c = (x, y)$, we define the \emph{potential} of~$c$ to be $\phi(c) = x + y$. Let~$s$ be the bottommost square in $\mathcal{C} \setminus \mathcal{C}_2$ whose cell has maximum potential, and let~$e$ be the topmost empty grid cell with minimum potential that is occupied in $\mathcal{C}_2$. We iteratively move~$s$ to~$e$ in~$\mathcal{C}$ until $\mathcal{C} = \mathcal{C}_2$.
    We first show that~$\mathcal{C}$ remains $xy$-monotone.
    Removing~$s$ cannot break this property:  
    by definition of~$\phi$ and since $\mathcal{C}_2$ is $xy$-monotone, $s$ does not have \dirN-, \dirNE-, or \dirE-neighbors in~$\mathcal{C}$. 
    Moreover, if it has a \dirNW-neighbor then, by $xy$-monotonicity of~$\mathcal{C}$, it has a \dirW-neighbor too. 
    Similarly, adding a square in $e$ maintains $xy$-monotonicity. 
    By the definition of $\phi$ and since $\mathcal{C}_2$ is $xy$-monotone, the cells neighboring $e$ in the \dirS, \dirSW, and \dirW directions must be occupied if they are inside the bounding box of~$\mathcal{C}$. 
    Moreover, $xy$-monotonicity of~$\mathcal{C}$ guarantees that $e$ does not have \dirN-, \dirNE-, or \dirE-neighbors.
    
    %By definition, the width and height of the bounding box of the canonical configuration are upperbounded by the length of the largest side of the bounding box of $\mathcal{C}$. 
    In every step we move a square from a position occupied in $\mathcal{C}_1$ to a position occupied in~$\mathcal{C}_2$, hence the perimeter of the bounding box of configuration $\mathcal{C}$ is $O(P + P') = O(\bar{P})$. 
    Since moving along the boundary of~$\mathcal{C}$ takes at most $O(\bar{P})$ moves 
    and no square is moved more than once, in total it takes $O(\bar{P}n)$ moves to reconfigure. 
\end{proof}

\section{Illegal moves in {MSDP}}
\label{sec:bugs}

The checks implemented in our tool detected illegal moves executed by the implementation of MSDP available online.\cref{foot:MSlink}
These illegal moves come in different flavors (see Figure~\ref{fig:bugs}): 
\begin{enumerate}[(a)] 
\item illegal convex transitions around an empty grid position,
\item truncated convex transitions (can be also thought of as illegal slides) in which the final position 
is the intermediate cell of a convex transition, 
\item illegal slides in which one of the two cells along which the square slides is empty, and
\item illegal move sequences in which a convex transition is broken into two different moves and a new move starts on the illegal intermediate state. This can be also thought of as an illegal slide.  
\end{enumerate}

The illustrations in Figure~\ref{fig:bugs} are screenshots from the online tool. 
They can be reproduced by drawing the configuration on the left in the online tool and trying to reconfigure it into any other connected configuration with the same number of squares.

\begin{figure}[t]
    \centering
    
    \begin{minipage}[t]{0.48\textwidth}
        \centering
        \includegraphics[scale=0.4]{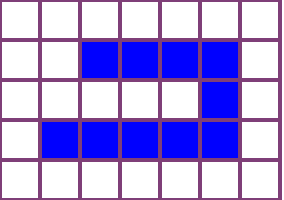}
        \hspace{1em}
        \includegraphics[scale=0.4]{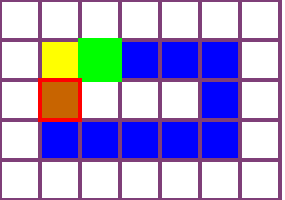}
        \subcaption{Illegal convex transition around an empty cell.}
        \label{fig:invalid-ct}
    \end{minipage}
    \hfill
    \begin{minipage}[t]{0.48\textwidth}
        \centering
        \includegraphics[scale=0.4]{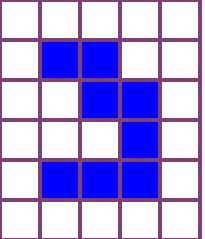}
        \hspace{1em}
        \includegraphics[scale=0.4]{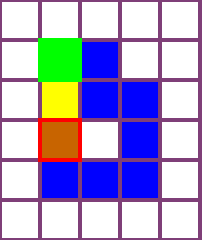}
        \subcaption{Illegal slide/truncated convex transition. 
        The red cell is an intermediate position in a valid convex transition, 
        but the move sequence cannot end here.}
        \label{fig:invalid-truncated}
    \end{minipage}
    
    \begin{minipage}[t]{0.48\textwidth}
        \centering
        \includegraphics[scale=0.4]{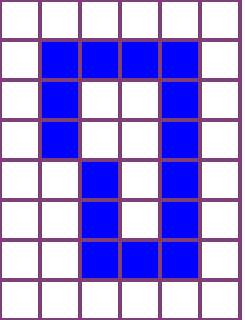}
        \hspace{1em}
        \includegraphics[scale=0.4]{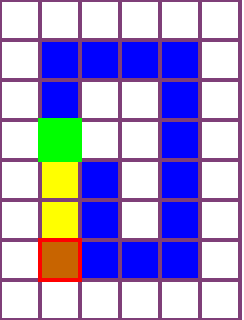}
        \subcaption{Illegal slide: the green square cannot directly slide to the adjacent yellow cell.}
        \label{fig:invalid-slide}
    \end{minipage}
    \hfill
    \begin{minipage}[t]{0.48\textwidth}
        \centering
        \includegraphics[scale=0.4]{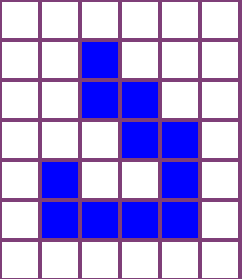}
        \hspace{1em}
        \includegraphics[scale=0.4]{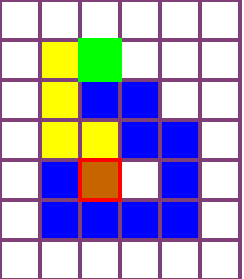}
        \subcaption{Illegal move sequence triggered by convex transition executed as two moves.}
        \label{fig:invalid-sequence}
    \end{minipage}
    \caption{Examples of illegal move sequences found in the implementation of MSDP.\cref{foot:MSlink} In each figure, the source configuration is shown on the left. On the right, the square marked in green is moved to the position marked in red by moving through the positions marked in yellow. The illustrations are screenshots from the online tool.}
    \label{fig:bugs}
\end{figure}

To remove illegal moves, we adapted the movement sequences of MSDP as follows. First of all, we removed all positions which correspond to the intermediate cell occupied in a (valid) convex transition. 
We then try to find valid movement sub-sequences connecting any two consecutive positions. 
However, some positions of the movement sequence given by the MSDP implementation are simply not reachable by a valid sequence of moves, such as the yellow cell in Figure~\ref{fig:bugs}a.
In those cases, we delete these unreachable positions. 
All the remaining positions of the movement sequence can be connected using valid sub-sequences; we find those by computing the shortest sequence of moves between subsequent positions. Note that in all cases we encountered, the end position of a move sequence was always reachable from the start position but the valid path can require $\Omega(n)$ more moves. 

%\section{NP-hardness for pivoting hexagons} 
%\label{app:hex}

\end{document}